\documentclass[11pt,a4paper]{article}
\usepackage[utf8]{inputenc}
\usepackage[T1]{fontenc}
\usepackage{lmodern}
\usepackage{microtype}
\usepackage[pdfborder={0 0 0}]{hyperref}
\usepackage{tikz}

\usepackage{color}
\usepackage{marginnote} 
\usepackage{ifthen} 
\usepackage{paralist} 
\usepackage{graphicx}

\usepackage[tbtags]{amsmath} 
\usepackage{amssymb}
\usepackage{amsthm}
\usepackage{dsfont} 
\usepackage{mathtools} 

\usepackage[totalwidth=16.5cm,totalheight=23.5cm]{geometry}

\usepackage[normalem]{ulem} 
\renewcommand{\sout}[1]{\unskip}

\usepackage[all]{xy}

\def\Ad{\mathrm{Ad}}

\newcommand{\inv}[0]{{-1}}

\newcommand{\oo}[0]{\otimes}
\newcommand{\ee}[0]{\epsilon}

\makeatletter
\newtheorem*{rep@theorem}{\rep@title}
\newcommand{\newreptheorem}[2]{%
\newenvironment{rep#1}[1]{%
 \def\rep@title{#2 \ref{##1}}%
 \begin{rep@theorem}}%
 {\end{rep@theorem}}}
\makeatother

\newtheorem{theorem}{Theorem}[section]

\newtheorem{proposition}[theorem]{Proposition}

\newreptheorem{theorem}{Theorem}
\newreptheorem{lemma}{Lemma}
\newreptheorem{corollary}{Corollary}
\newreptheorem{proposition}{Proposition}
\newreptheorem{claim}{Claim}
\newreptheorem{conjecture}{Conjecture}
\newreptheorem{problem}{Problem}

\newcommand{\issue}[2][]{%
  \textcolor{red}{#2}%
  \marginnote{\textbf{\textcolor{red}{issue}}%
    \ifthenelse{\equal{#1}{}}{}{: #1}}}

\def\Hom              {{\rm Hom}}

\def\C                {\mathbb C}

\def\id               {{\rm id}}

\long\def\labl#1      {\label{#1}\ee}

\def\R                {\mathbb R}

\def\Im               {{\rm Im}}

\def\Re               {{\rm Re}}
\def\ker               {{\rm Ker}}
\def\Ker               {{\rm Ker}}

\def\tr               {{\rm Tr}}

\def\Z                {\mathbb Z}

\newcommand{\II}{I\kern -0.7ex I}
\newcommand{\III}{I\kern -0.7ex I\kern -0.7ex I}

\newcommand{\cC}{{\mathcal C}}
\newcommand{\cH}{{\mathcal H}}

\newcommand{\cL}{{\mathcal L}}
\newcommand{\cM}{{\mathcal M}}
\newcommand{\cP}{{\mathcal P}}

\newcommand{\cG}{{\mathcal G}}
\newcommand{\cR}{{\mathcal R}}
\newcommand{\cT}{{\mathcal T}}

\newcommand{\bbH}{{\mathbb H}}
\newcommand{\bbP}{{\mathbb P}}

\newcommand{\bbS}{{\mathbb S}}
\newcommand{\pr}{{\rm pr}}

\newcommand{\Diff}{{\rm Diff}}

\newcommand{\Mod}{{\rm Mod}}

\newcommand{\Hyp}{{\rm Hyp}}
\newcommand{\Rep}{{\rm Rep}}

\newcommand{\Isom}{{\rm Isom}}

\newcommand{\Ann}{{\rm Ann}}

\newcommand{\Li}{{\rm Li}}

\newcommand{\Ein}{{\rm Ein}}

\title{Generalised shear coordinates on the moduli spaces of three-dimensional spacetimes}
\author{Catherine Meusburger\footnote{{\tt catherine.meusburger@math.uni-erlangen.de}}~ and Carlos Scarinci\footnote{{\tt scarinci@math.fau.de}} \\ \\ {\it Department Mathematik, Friedrich-Alexander Universit\"at  Erlangen-N\"urnberg,} \\ {\it Cauerstra\ss e 11, 91058 Erlangen, Germany}}

\date{}

\begin{document}

\maketitle

\begin{abstract}
We introduce 
coordinates on the moduli spaces of  maximal globally hyperbolic constant curvature 3d spacetimes with cusped Cauchy surfaces $S$. They are derived from the parametrisation of the moduli spaces by the bundle of measured geodesic laminations over Teichm\"uller space of $S$ and can be viewed as analytic continuations of the shear coordinates on Teichm\"uller space. In terms of these coordinates the gravitational symplectic structure takes a particularly simple form, which resembles the Weil-Petersson symplectic structure in shear coordinates, and is closely related to the cotangent bundle of Teichm\"uller space. We then consider the mapping class group action on the moduli spaces and show that it preserves the gravitational symplectic structure. This defines three distinct mapping class group actions on the cotangent bundle of Teichm\"uller space, corresponding to different values of the curvature.
\end{abstract}

\section{Introduction}\label{sec:intro}

\subsubsection*{Moduli spaces of constant curvature spacetimes}

Moduli spaces of three-dimensional constant curvature spacetimes classify the diffeomorphism classes of solutions of the Einstein equations on three-dimensional manifolds. As the Ricci curvature of a three-dimensional manifold determines its Riemann curvature tensor, all solutions of the Einstein equations with vanishing stress-energy tensor have constant curvature $\Lambda$, where $\Lambda$ is the cosmological constant.  
 This implies in particular that three-dimensional Einstein spacetimes, that is, Ricci constant spacetimes,  are all locally isometric 
  to one of three model Lorentzian geometries:  three-dimensional Minkowski space,  anti-de Sitter space or  de Sitter space for, respectively, $\Lambda=0$, $\Lambda=-1$ and $\Lambda=1$. 
  
  As a consequence, three-dimensional Einstein spacetimes can be classified completely under certain additional assumptions on their causality structure. 
  This yields a  classification of maximal globally hyperbolic (MGH) three-dimensional Einstein spacetimes of topology $\R\times S$, where $S$ is an 
   orientable surface, possibly with punctures.  Remarkably, these three-dimensional structures are completely characterised in terms of two-dimensional structures and their associated moduli spaces, in particular those related to hyperbolic geometry and Teichm\"uller theory.

Important results in this respect are the work by  Mess \cite{Mess2007} and Scannell \cite{Scannell1999}, where the diffeomorphism classes of MGH Einstein spacetimes with a compact Cauchy surface $S$ are classified by the cotangent bundle over Teichm\"uller space $T^*\cT(S)$ for $\Lambda=0$, two copies of Teichm\"uller space $\cT(S)\times\cT(S)$ for $\Lambda=-1$ and the space of complex projective structures $\cC\cP(S)$ for $\Lambda=1$. An explicit geometric construction of such spacetimes in terms of domains of dependence in the corresponding model spacetimes was later given by Benedetti and Bonsante \cite{Benedetti2009}. These domains of dependence are described in terms of earthquakes and grafting along measured laminations, and 
 the solutions for different values of curvature are related via so-called canonical Wick rotations and rescalings. This allows for a clear geometrical description of the moduli space of three-dimensional MGH Einstein spacetimes as 
the bundle $\cM\cL(S)$  of measured geodesic laminations over Teichm\"uller space and generalises the well-known description of the moduli space of three-dimensional hyperbolic manifolds by Thurston \cite{Thurston1980}.

The fact that any three-dimensional MGH Einstein spacetime is locally isometric to one of the model Lorentzian geometries also gives rise to a classification of such spacetimes in terms of conjugacy classes of group homomorphism $\pi_1(S)\to G_\Lambda$, where $G_\Lambda$ is the isometry group of the model spacetime.
This identifies their moduli spaces with a certain subspace of the corresponding representation variety $\Hom(\pi_1(S),G_\Lambda)/G_\Lambda$ or, equivalently, of the moduli space of flat $G_\Lambda$-connections on $S$. This can be viewed as a direct generalisation of the realisation of  Teichm\"uller space $\cT(S)$ as 
a connected component of the representation variety $\Hom(\pi_1(S),\mathrm{PSL}(2,\R))/\mathrm{PSL}(2,\R)$. 

From a physics perspective, this is related to the Chern-Simons formulation of three dimensional gravity developed by Achucarro-Townsend \cite{Achucarro1986} and Witten \cite{Witten1988/89} which describes three-dimensional gravity  as a $G_\Lambda$-Chern-Simons theory on $\R\times S$.  The representation variety $\Hom(\pi_1(S), G_\Lambda)/G_\Lambda$ then corresponds to the gauge invariant phase space of this Chern-Simons theory. The fact that only a certain subset of this phase space corresponds to gravity follows from the non-degeneracy of the metric, which imposes restrictions on the Chern-Simons connection.

\subsubsection*{Three-dimensional spacetimes, Teichm\"uller theory and quantisation}

This relation between three-dimensional Einstein geometry,  two-dimensional hyperbolic geometry and flat connections makes the moduli spaces of MGH Einstein spacetimes an interesting research topic from different perspectives. On one hand, Teichm\"uller theory is a very rich and well-developed theory with deep connections with complex analysis, algebraic topology, low dimensional topology and geometry and also symplectic geometry, to cite but a few, see \cite{Papadopoulos2007,Hubbard2006}. The study of generalisations of Teichm\"uller space and its structures may thus lead to new insights and techniques for a wide variety of topics in mathematics.

From a physics viewpoint, moduli spaces of MGH Einstein spacetimes are interesting since they are the diffeomorphism invariant phase space of gravity in three dimensions. Three-dimensional gravity plays an important role as a toy model for the quantisation of general relativity which allows one to investigate conceptual questions of quantum gravity and to develop new approaches to quantisation, see \cite{Carlip1998} and references therein for an overview.
Moreover, the  quantisation of  moduli spaces of MGH Einstein spacetimes is also of mathematical interest due to their relation with the moduli spaces of flat connections and Chern-Simons gauge theory, which connects it directly to the construction of 
quantum invariants of three-manifolds and three-dimensional topological quantum field theory \cite{Witten1989a,Reshetikhin1991,Turaev1992,Barrett1996}.

While quantisation techniques for Chern-Simons theory are well-established for  compact, semisimple  gauge groups, the case of non-compact groups remains a challenge due to their more complicated representation theory. A first step in the generalisation of the definition of quantum invariants associated with non-compact Lie groups such as the group  $\mathrm{PSL}(2,\C)$ are the hyperbolic invariants of Baseilhac and Benedetti \cite{Baseilhac2004,Baseilhac2007}, which in turn are closely related with the theory of quantum Teichm\"uller spaces \cite{Kashaev1998,Fock1999} via the quantum dilogarithm of Faddeev-Kashaev \cite{Faddeev1994}, see also \cite{Bai2007,Bonahon2007}.

The challenges in  the quantisation of moduli spaces of flat connections for non-compact groups and the close relation with Teichm\"uller theory thus suggest to approach the quantisation of the gravitational moduli spaces by applying or generalising results from quantum Teichm\"uller theory. In particular, this provides a strong motivation  to  generalise the description of Teichm\"uller space in terms of Thurston's shear coordinates \cite{Thurston1998,Bonahon1996,Fock1997} to the context of the gravitational moduli spaces.

\subsubsection*{3d Einstein spacetimes in terms of generalised shear coordinates}

The shear coordinate parametrisation of Teichm\"uller space has a direct geometrical interpretation and yields a simple description of the Weil-Petersson symplectic structure and the action of the mapping class group, which play a central role in the derivation of quantum Teichm\"uller theory. The present article introduces a set of coordinates on the moduli spaces $\cG\cH_\Lambda(\R\times S)$ of MGH Einstein spacetimes that can be viewed as a natural generalisation of Thurston's shear coordinates to the three-dimensional Lorentzian context. Similarly to the two-dimensional case, these generalised shear coordinates are defined by means of an ideal triangulation of a cusped surface and are obtained via the identification of the moduli space of MGH Einstein spacetimes with the bundle of measured geodesic laminations over Teichm\"uller space. They have a direct geometrical interpretation in terms of shearing and bending of hyperbolic structures along the ideal edges of the triangulation and allow one to directly 
determine  the associated group homomorphisms $\pi_1(S)\to G_\Lambda$. This should also be compared with the work of Bonahon \cite{Bonahon1996} where 
complex measured laminations are used to define complex-valued shear coordinates on the moduli space of hyperbolic 3-manifold.

From a more algebraic perspective, the coordinates introduced in this article  can be understood as analytic continuations of the shear coordinates in Teichm\"uller space with values in a two-dimensional commutative real algebra $R_\Lambda$, which coincides with the complex numbers for $\Lambda=1$, with the dual numbers for $\Lambda=0$ and with the split complex numbers for $\Lambda=-1$. 
This algebra also provides a unified description of the isometry groups $G_\Lambda$ and their Lie algebras in terms of matrices with entries in $R_\Lambda$. The $R_\Lambda$-valued shear coordinate then arise from the generalisation of the shear coordinate parametrisation of the holonomy representations.

The gravitational symplectic structure on $\cG\cH_\Lambda(\R\times S)$, which is the restriction of Goldman's symplectic form \cite{Goldman1984} on $\Hom(\pi_1(S), G_\Lambda)/G_\Lambda$, takes a particularly simple form in terms of these generalised shear coordinates. It is purely combinatorial and resembles the Weil-Petersson symplectic structure on Teichm\"uller space in shear coordinates. Moreover, for all values of $\Lambda$, it is also closely related to the cotangent bundle $T^*\cT(S)$ of Teichm\"uller space. In fact, the spaces $\cG\cH_\Lambda(\R\times S)$, for different values of $\Lambda$, are show to be isomorphic as symplectic manifolds to $T^*\cT(S)$.

The second part of the article investigates the action of the mapping class group $\Mod(S)$ on the moduli spaces of MGH Einstein spacetimes. Concrete expressions for mapping class group action on the generalised shear coordinates are derived in terms of Whitehead moves, which can be viewed as an analytic continuation of the corresponding expressions for shear coordinates on $\cT(S)$. The Whitehead moves are shown to preserve the gravitational symplectic structure and thus induce a symplectic $\Mod(S)$-action on $\cG\cH_\Lambda(\R\times S)$, for all values of $\Lambda$. This, in turn, induces three distinct symplectic actions on the cotangent bundle $T^*\cT(S)$. These results are achieved via a simple decomposition of the Whitehead moves into terms generated via the symplectic structure and linear terms which implement the combinatorial transformation of the Poisson structure under Whitehead moves.
 
The Hamiltonians generating the non-linear part of the Whitehead moves are related to the ``imaginary part'' of the dilogarithm of the associated edge coordinate, which is known to be related to the hyperbolic volume of ideal tetrahedra. We show that this is a direct generalisation of a corresponding result for Teichm\"uller space, in which the relevant Hamiltonian is the real dilogarithm. This gives a clear motivation for the appearance of the quantum dilogarithm in the quantum theory and  suggests that a quantisation of these moduli spaces could be used to define three-manifold invariants similar to those in  \cite{Baseilhac2004,Baseilhac2007}.

\section{Teichm\"uller theory and shear coordinates}
\label{sec:background}

In this section we summarise the relevant background on Teichm\"uller theory and Thurston's shear coordinates. We refer the reader to \cite{Papadopoulos2007} for a general introductory overview, but discuss the aspects  relevant to this article in detail to make it self-contained. In the following, $S$ denotes an orientable surface of genus $g$
 with $s$ punctures. We assume that the surface  contains at least one puncture  ($s>0$)  and  that its universal cover is isometric to the hyperbolic plane ($2g-2+s>0$).  

\subsection{Teichm\"uller space}
\label{subsec:back_teichmueller}

\subsubsection*{The Riemann moduli space and Teichm\"uller space}

The Riemann moduli space $\cR(S)$ parametrises diffeomorphism classes of, both,  hyperbolic and  conformal structures on a given topological surface $S$. In the case of punctured Riemann surfaces ($s>0$) there are different versions of this space, depending on the boundary conditions imposed on the hyperbolic metrics near each puncture \cite{Fock1997,Penner1993}. In this article, we consider the Riemann moduli space of cusped hyperbolic structures. We thus define $\cR(S)$ as the space of finite area complete hyperbolic metrics on $S$ modulo orientation preserving diffeomorphisms
$$\cR(S)=\Hyp(S)/\Diff^+(S).$$
The universal cover of the Riemann moduli space $\cR(S)$ is the Teichm\"uller space $\cT(S)$, which is the space of finite area complete hyperbolic metrics on $S$ modulo diffeomorphisms  isotopic to the identity
$$\cT(S)=\Hyp(S)/\Diff_{0}(S).$$
The group of deck transformations of the covering $\cT(S)\to\cR(S)$ is
 the mapping class group of $S$, which is given as the 
 quotient of the group $\Diff^+(S)$ of orientation preserving diffeomorphisms of $S$ by its normal subgroup $\Diff_0(S)$ of diffeomorphisms isotopic to the identity
$$\Mod(S)=\Diff^+(S)/\Diff_0(S).$$

\subsubsection*{
Symplectic structure}

Teichm\"uller space $\cT(S)$ carries a canonical symplectic structure. This  can be understood via its relation to  the $\mathrm{PSL}(2,\R)$-representation variety, which consists of conjugacy classes of  group homomorphisms  $\pi_1(S)\to \mathrm{PSL}(2,\R)$
$$\Rep(S,\mathrm{PSL}(2,\R))=\Hom(\pi_1(S),\mathrm{PSL}(2,\R))/\mathrm{PSL}(2,\R).$$
Via the uniformisation theorem, points $h\in\cT(S)$ are in one-to-one correspondence with Fuchsian representations $\rho\in\Rep(S,\mathrm{PSL}(2,\R))$ with fixed parabolic conjugacy classes around each puncture. Thus Teichm\"uller space can be viewed as a connected component of the representation variety
$\Rep_\cC(S,\mathrm{PSL}(2,\R))$,
where the index $\cC$  indicates the restriction to fixed parabolic conjugacy classes at the punctures.

For any Lie group $G$, the representation variety $\Rep(S,G)$ 
coincides with the moduli space of flat $G$-connections on $S$ and  carries a canonical non-degenerate closed two-form, the Atiyah-Bott-Goldman symplectic form  \cite{Atiyah1983,Goldman1984}, which is determined by the choice of an Ad-invariant symmetric bilinear form on $\mathfrak g=\text{Lie}\, G$. From a physics viewpoint, this moduli space is  the (gauge invariant) phase space of Chern-Simons theory on the three-manifold $\R\times S$. The $\Ad$-invariant symmetric bilinear form which characterises the Atiyah-Bott-Goldman symplectic structure enters in the definition of the Chern-Simons action and this  symplectic structure can be seen as the associated physical symplectic structure on the phase space \cite{Atiyah1983}.

For the group $G=\mathrm{PSL}(2,\R)$, the choice of an Ad-invariant symmetric bilinear form is unique up to rescaling, and it was shown by Goldman \cite{Goldman1984} that the restriction of the associated symplectic structure to Teichm\"uller space  $\cT(S)\subset \Rep(S,\mathrm{PSL}(2,\R))$ coincides with
the Weil-Petersson Poisson structure on $\cT(S)$.

\subsubsection*{Measured laminations and earthquakes}\label{subsec:back_lamin}
Measured geodesic laminations on two-dimensional hyperbolic surfaces and the associated operation of earthquake play a prominent role in three-dimensional hyperbolic geometry \cite{Thurston1980,Epstein1987,Bonahon1996,McMullen1998} as well as in three-dimensional Lorentzian geometry \cite{Mess2007,Scannell1999,Benedetti2009,Meusburger2007}.  
Measured geodesic laminations can be viewed as generalisations of weighted simple closed geodesics and earthquakes are defined via cutting and gluing operations along such geodesics. More precisely, a measured geodesic laminations on $S$ associated to a hyperbolic metric
$h\in\cT(S)$  is a pair $(\lambda,\mu)$ formed by:
\begin{enumerate}
\item a closed subset $\lambda\subset S$ that is foliated by disjoint non-self-intersecting complete geodesics, called the leaves of the lamination,
which cannot be contracted to punctures;
\item a positive measure $\mu$ on the set of arcs transverse to the leaves, which is invariant under homotopy through transverse arcs and additive under concatenation of arcs.
\end{enumerate}
We denote by $\cM\cL(S)$ the (total space of the) bundle of measured geodesic laminations over $\cT(S)$, considered up to isotopy.

Note that although the definition of  a  measured geodesic lamination makes use
 of a reference hyperbolic metric $h\in\cT(S)$,  there is a canonical identification between measured laminations defined with respect to any pair of hyperbolic metrics $h,h'\in\cT(S)$. This follows since any geodesic for a metric $h$ can be smoothly deformed to a geodesic for a metric $h'$, which gives a global identification between the fibres of the (trivial) bundle $\cM\cL(S)$.
In the following, however, such an identification of measured geodesic laminations for different metrics will not be convenient, and we shall not rely on any choice of trivialisation.

An earthquake is an operation $Eq:\cM\cL(S)\to\cT(S)$ that associates  to each point $h\in \cT(S)$ and each measured geodesic lamination $\lambda\in\cM\cL_h(S)$ another point $Eq^\lambda(h)\in\cT(S)$, called the earthquake of $h$ along $\lambda$. If $\lambda$ is a simple closed geodesic with associated weight $\mu\in \R_+$, the hyperbolic metric $Eq^\lambda(h)$ is obtained by cutting the hyperbolic surface determined by $h$ along $\lambda$ and gluing the pieces back together after applying a  (right) twist by $2\pi\mu$.  

A more explicit description of the earthquake operation can be given in terms of the associated Fuchsian representations of the fundamental group of $S$. For a point  $h\in\cT(S)$ denote by $\rho:\pi_1(S)\to\mathrm{PSL}(2,\R)$  the corresponding Fuchsian representation, which is  determined uniquely up to conjugation via the uniformisation theorem.
The earthquake $Eq^\lambda(h)$ of $h$ along $\lambda$ is then determined from  a new Fuchsian representation 
 $\rho^\lambda:\pi_1(S)\to\mathrm{PSL}(2,\R)$ constructed as follows. 
 First consider  the lift $\tilde\lambda\in\cM\cL(\bbH^2)$ of $\lambda$ to the universal cover of $(S,h)$. Each leaf $\tilde l$ of $\tilde\lambda$ is then a complete geodesic in the hyperbolic plane, and hence for each point $p\in\tilde l$ there is a unique
 hyperbolic isometry $A_p\in \mathrm{PSL}(2,\R)$ mapping the imaginary axis to $\tilde l$,  the point $i$ to $p$ and preserving the orientation. This allows one to associate to the representation $\rho:\pi_1(S)\to\mathrm{PSL}(2,\R)$ and the lamination $\lambda$ a $\rho$-cocycle $Z_E^\lambda:\pi_1(S)\to\mathrm{PSL}(2,\R)$ defined by
\begin{align}\label{eq:zedef}
Z_E^\lambda(a)=\prod_{p\in \lambda\cap a}\Ad_{A_p}E(\mu_p)\qquad \forall a\in\pi_1(S),
\end{align}
where
\begin{align}\label{eq:coord_earth}
E(\mu)=\left(\begin{array}{cc} e^{\mu/2} & 0\\ 0 & e^{-\mu/2} \end{array}\right)\in\mathrm{PSL}(2,\R)
\end{align}
is the hyperbolic translation of length $\mu$ along the imaginary axis and $\mu_p=\epsilon_p(\lambda,a)\mu$ is the measure of $a$ at $p$ multiplied by the oriented intersection number between $\lambda$ and $a$. The fact that $Z^\lambda_E$ is a $\rho$-cocycle 
\begin{align*}
Z^\lambda_E(ab)=Z^\lambda_E(a)\Ad_{\rho(a)}Z^\lambda_E(b),\qquad\forall a,b\in\pi_1(S),
\end{align*}
ensures that one obtains a  representation $\rho^\lambda:\pi_1(S)\to\mathrm{PSL}(2,\R)$ by setting
$$\rho^\lambda(a)=Z^\lambda_E(a)\rho(a),\qquad \forall a\in\pi_1(S).$$
This representation is again Fuchsian and therefore it determines a unique hyperbolic metric $Eq^\lambda(h)\in\cT(S)$, the earthquake of $h$ along $\lambda$.

It was show by Thurston, see \cite{Kerckhoff1983} for the proof,
that any two  hyperbolic metrics are related via an earthquake. More precisely, Thurston's earthquake theorem states that for any pair of hyperbolic metrics $h,h'\in\cT(S)$ there exists a unique measured geodesic lamination $\lambda\in\cM\cL_h(S)$ such that $Eq^\lambda(h)=h'$. In particular, for a given point $h\in\cT(S)$ there is a bijection $\cM\cL_h(S)\to\cT(S)$ between the fibre of measured geodesic laminations over $h$ and Teichm\"uller space.

\subsection{Shear coordinates on Teichm\"uller space}
\subsubsection*{Definition of shear coordinates on $\cT(S)$}
A very effective tool in Teichm\"uller theory, in particular in  the study of its Poisson geometry and subsequent quantisation, is a special set of global coordinates on $\cT(S)$ associated with ideal triangulations of the punctured surface $S$. These coordinates  were first introduced by Thurston \cite{Thurston1998} and  further developed  by Bonahon \cite{Bonahon1996} and Fock \cite{Fock1997} and have a direct geometrical interpretation. They measure the hyperbolic displacement, or shear, of adjacent ideal triangles ---  see also \cite{Penner1987} for a related interpretation in terms of distances between horocycles at each puncture.

Let $\tau$ be an ideal triangulation of the surface $S$, that is, a triangulation of $S$ whose set of vertices coincides with the set of punctures of the surface. Note that being ideal is a rather restrictive condition on the triangulation, which determines the number of its vertices, edges and faces  uniquely. 
From the formula $v-e+f=2-2g$ for the Euler characteristic of $S$ together with the relations $2e=3f$ and $v=s$ it follows that the number of edges and faces of the ideal triangulation are given by
$e=6g-6+3s$ and $f=4g-4+2s$.

In the following we will also consider  the dual graph $\Gamma$ of  an ideal triangulation $\tau$.
Note that  for a general combinatorial graph the notion of face is not defined a priori. However, graphs dual to a triangulation of an oriented surface $S$ carry additional structure, namely a cyclic ordering of the incident edges at each vertex induced by the orientation of the underlying surface. 
A graph with such a cyclic ordering of the incident edges at each vertex is called a fat graph. The notion of face can then be defined as certain closed edge paths on $\Gamma$.
In the case of a trivalent fat graph dual to a triangulation, a face is a closed edge path of the graph which takes the ``same turn'', either  left or right, at each vertex and which does not pass through an edge twice in the same direction.

An ideal triangulation $\tau$ of a surface $S$ thus determines uniquely (up to isotopy) an embedded trivalent dual fat graph $\Gamma$. While different ideal triangulations of $S$ lead to different embedded fat graphs, the associated fat graphs are always related by sequences of Whitehead moves, see  Section \ref{subsec:back_mcg}.
Conversely, for a trivalent fat graph $\Gamma$ there is a  unique (up to diffeomorphism) oriented surface with an ideal  triangulation, which is obtained by gluing punctured discs along the faces of the graph. Graphs related by a sequence of Whitehead moves give rise to the same topological surface and punctured surfaces are therefore in one-to-one correspondence with trivalent fat graphs modulo Whitehead moves.

Given a surface $S$ and a corresponding embedded trivalent fat graph $\Gamma$, one defines
Thurston's shear coordinates on Teichm\"uller space $\cT(S)$  as follows. Denote  by $V(\Gamma)$, $E(\Gamma)$ and $F(\Gamma)$, respectively,  the  sets of  vertices, edges and faces of $\Gamma$. A point $h\in\cT(S)$ corresponds to an equivalence class of hyperbolic structures on $S$ and determines an ideal geodesic triangulation of $S$ dual to $\Gamma$. Each edge $\alpha\in E(\Gamma)$ corresponds to an ideal hyperbolic square on $S$ or, equivalently, to a $\pi_1(S)$-equivalence class of hyperbolic squares in the 
 universal cover. The shear coordinate 
 $x^\alpha=x^\alpha(h)$ assigned to $\alpha$ is then defined as the logarithm of the cross-ratio associate with the ideal square determined by $\alpha$. 
 
 Working with the upper-half plane model of the hyperbolic plane, we may normalise this ideal square in such a way that one of the two adjacent triangles has vertices at $-1,0,\infty$ and the other at $\infty,0,t$ as shown in Figure \ref{fig:cross_ratio}. The cross-ratio is then given by the coordinate $t\in\R$ of the fourth vertex and the shear coordinate takes the form
$x^\alpha(h)=\log t$.
\begin{figure}[!ht]
\includegraphics[scale=0.34]{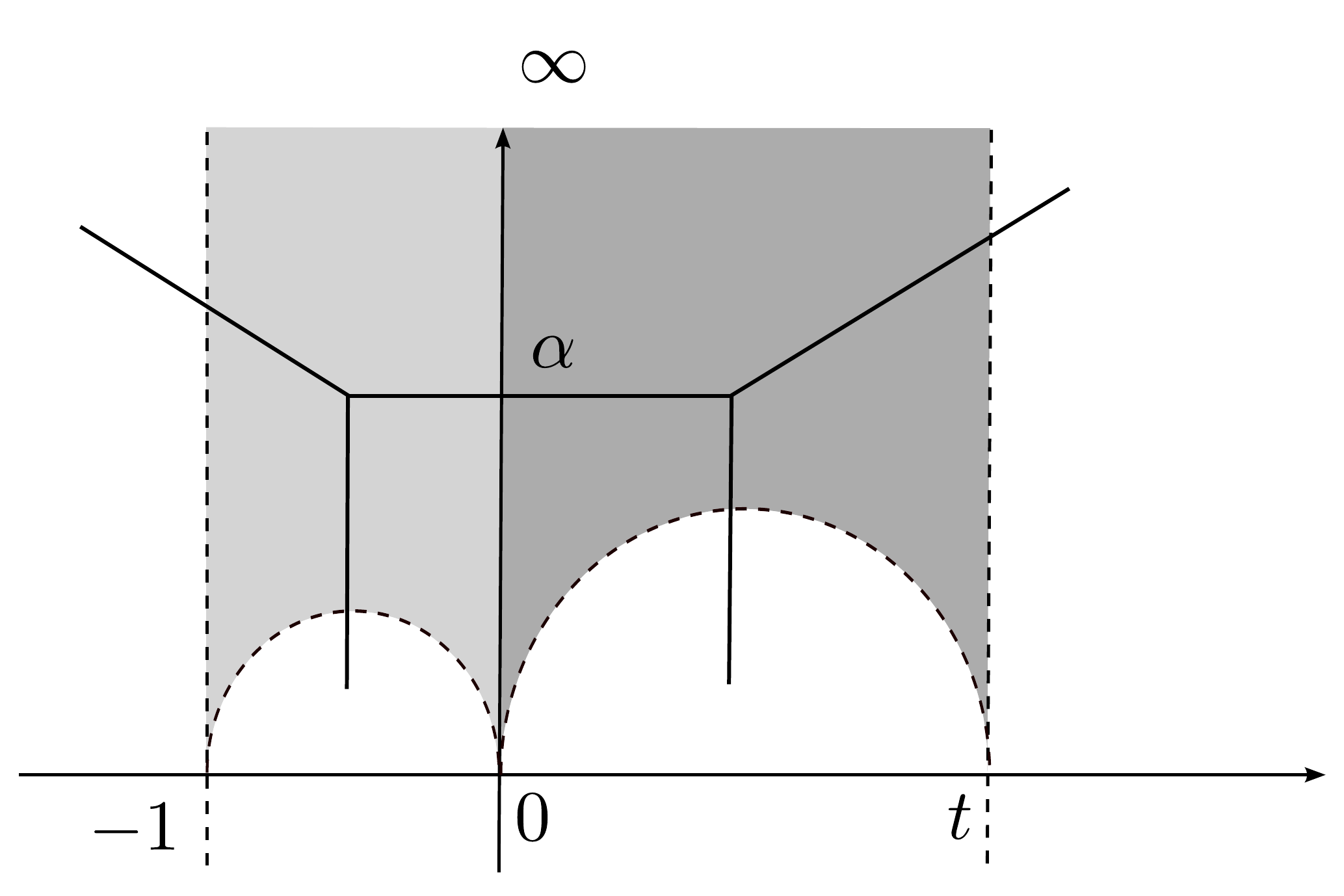}
\centering
\caption{Ideal square determined by the edge $\alpha$ of the fat graph dual to an ideal triangulation.}
\label{fig:cross_ratio}
\end{figure}

The shear coordinate $x^\alpha$ has a simple 
 geometric interpretation  as the shear between the  two adjacent ideal triangles. This follows from the fact that the triangle with vertices $\infty,0,t$ is obtained from a reference triangle with vertices $0,1,\infty$ by applying the hyperbolic transformation $E(x^\alpha)$  in \eqref{eq:coord_earth}. This transformation preserves the imaginary axis, which is  the lift of the ideal geodesic dual to the edge $\alpha$,  and the signed hyperbolic distance between a point on the imaginary axis and its image is exactly $x^\alpha$. This allows one to  interpret the coordinate $x^\alpha$ as the weight parameter for an earthquake along the imaginary axis, as shown in Figure \ref{fig:shear_triangle}.
\begin{figure}[!ht]
\includegraphics[scale=0.35]{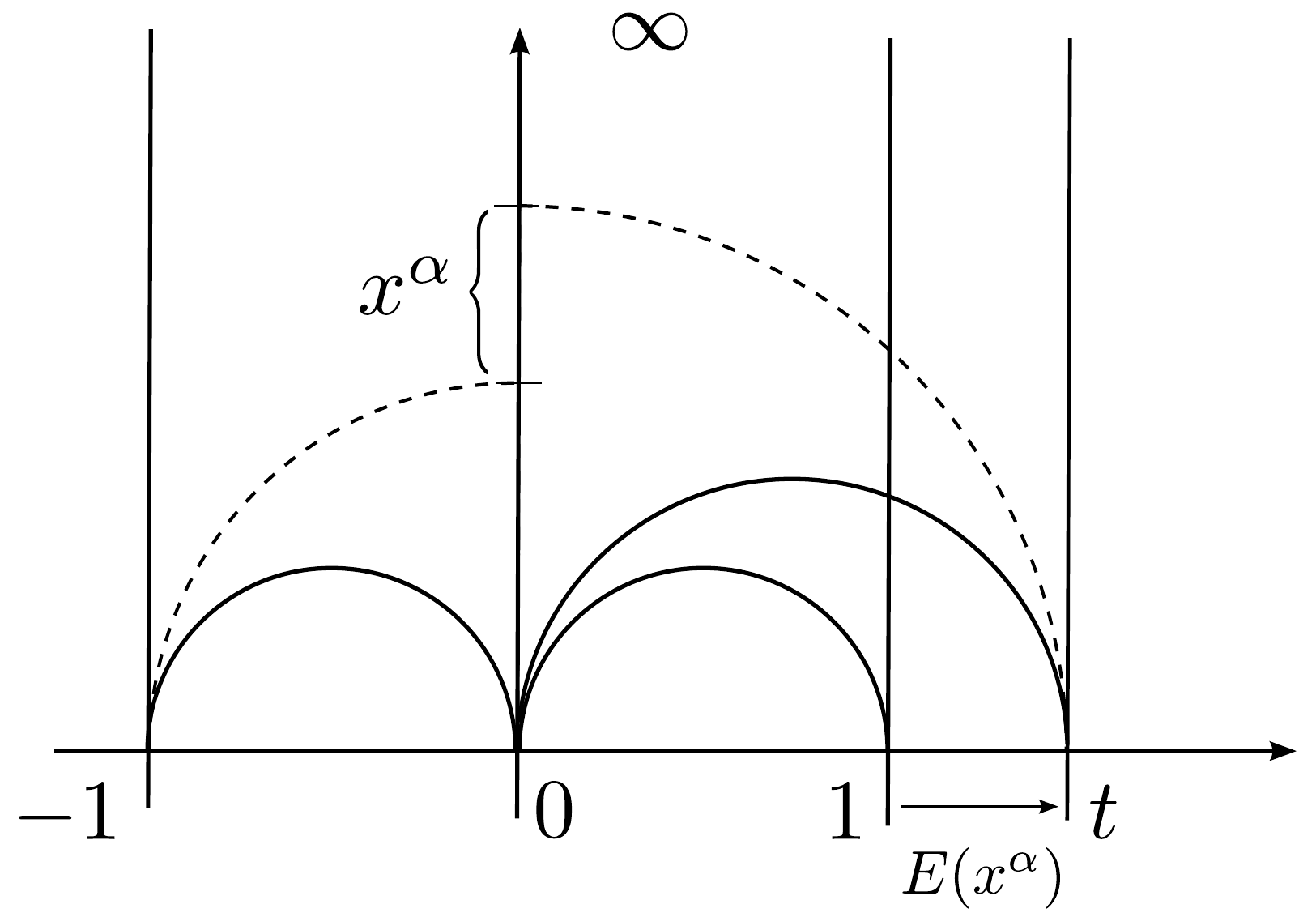}
\centering
\caption{Earthquake with weight $x^\alpha$ along the imaginary axis.}
\label{fig:shear_triangle}
\end{figure}

Besides their rather natural interpretation, an important feature of shear coordinates is their relation to the holonomies of simple closed curves on $S$, which are the images of elements of $\pi_1(S)$ under the associated Fuchsian representation $\rho:\pi_1(S)\to \mathrm{PSL}(2,\R)$. For a given vertex $v\in V(\Gamma)$, any 
simple closed curve $a$ on $S$ is homotopic to a unique closed edge path along the fat graph $\Gamma$ starting and ending at $v$. Such a closed edge path corresponds to a
sequence $(\alpha_1,...,\alpha_n)$ of edges of $\Gamma$, and the holonomy of $a$ is given by 
\begin{align}\label{eq:hol}
\rho(a)=P^a_nE(x^{\alpha_n})\cdots P^a_1E(x^{\alpha_1}),
\end{align}
where $E(x^\alpha)$  is given by \eqref{eq:coord_earth} and $P^a_k$ is either
\begin{align*}
L=\left(\begin{matrix} 1 & 1 \cr 0 & 1 \end{matrix}\right)\quad\text{ or }\quad R=\left(\begin{matrix} 1 & 0 \cr 1 & 1\end{matrix}\right)
\end{align*}
depending on whether $\alpha_{k}$ comes before or after $\alpha_{k+1}$ with respect to the ordering of incident edges at their common vertex\footnote{Note that this parametrisation of the holonomies is equivalent to the one in \cite{Fock1999} although the matrices $L$ and $R$ are different from the matrices $L$ and $R$ used there.}.  The former corresponds to  a left turn and the latter to  a right turn at the vertex between  $\alpha_k$ and  $\alpha_{k+1}$. 
Note that this prescription determines the holonomies only up to the choice of a basepoint, and different choices give rise to holonomies that are related by conjugation.

As they correspond to deck transformations in the universal cover, the holonomies encode key geometric properties of the hyperbolic surface $S$. For instance, any element $a\in \pi_1(S)$ is homotopic to a unique closed geodesic of $S$, and its geodesic length $l(a)$ is given by the trace of the associated holonomy 
\begin{align*}
\tr\rho(a)=2\cosh({l(a)}/ 2).
\end{align*}
The parametrisation \eqref{eq:hol} of the holonomies  thus gives rise to a simple description of the geometric properties of the hyperbolic surface in terms of the shear coordinates $x^\alpha$.

It is also directly apparent from \eqref{eq:hol} that the shear coordinates $x^\alpha:\cT(S)\to\R$ cannot be all independent but  must satisfy certain constraints associated with the  boundary conditions at the punctures. By definition, each face of the graph $\Gamma$ corresponds to a puncture and the  holonomy of the associated edge path must be parabolic. On the other hand, for an edge path $(\alpha_1,...,\alpha_n)$ around a face, the matrices $P^a_k$ in \eqref{eq:hol} are necessarily all equal to  $L$ or all equal to  $R$. A simple computation then shows that the  parabolicity condition is equivalent to imposing the constraint that the sum of shear coordinates associated to the edges of a given face identically vanishes. More explicitly, for each face $i\in F(\Gamma)$ we have
\begin{align}\label{eq:teich_const}
c^i(x)=\sum_{\alpha\in E(\Gamma)}\theta^i{}_\alpha x^\alpha=0,
\end{align}
where $\theta^i{}_\alpha$ denotes the multiplicity of the edge $\alpha$ in the face $i$.
Note that these constraints are linear in the coordinates $x^\alpha$. Thus, denoting by $V$, $E$ and $F$ the number of vertices, edges and faces of $\Gamma$, one can interpret the constraints as a linear map $c:\R^E\to \R^F$,  which identifies $\cT(S)$ with a linear subspace of $\R^E$ of codimension $F$.
\begin{theorem}[Fock-Chekhov \cite{Fock1999}]
\label{thm:teich_shear}
The functions $x^\alpha:\cT(S)\to\R$ define an embedding 
$x:\cT(S)\hookrightarrow \R^E$
whose image is the kernel of the linear map $c:\R^E \to\R^F$ whose components are given by \eqref{eq:teich_const}.
\end{theorem}

\subsubsection*{Symplectic structure in shear coordinates}

Another  remarkable property of the shear coordinates is that they give rise to a very simple description of  the Weil-Petersson symplectic structure on $\cT(S)$ in terms of a Poisson structure on $\R^E$. This  Poisson structure is given in terms of  combinatorial constants  associated with the graph $\Gamma$ \cite{Fock1997,Penner1992}. Denoting by $\partial/\partial x^\alpha\in T\R^E$ the basis of coordinate vector fields on $\R^E$, one can characterise this  Poisson structure by the   Poisson bivector
\begin{align}\label{eq:wp_bivect}
\pi_{WP}
=\frac{1}{2}\sum_{\alpha\in E(\Gamma)}\frac \partial {\partial x^\alpha}\wedge\Big(\frac{\partial} {\partial x^\beta}-\frac{\partial} {\partial x^\gamma}+\frac{\partial} {\partial x^\delta}-\frac{\partial} {\partial x^\epsilon}\Big).
\end{align}
Here the sum is taken over all edges $\alpha\in E(\Gamma)$ and $\beta,\gamma,\delta,\epsilon$ are the incident edges at the source and target vertices of $\alpha$, ordered  as in Figure \ref{fig:adj_edges}. Note that this expression is also valid in the case where some of the edges $\beta, \gamma,\delta,\epsilon$ are equal.  
\begin{figure}[!ht]
\includegraphics[scale=0.34]{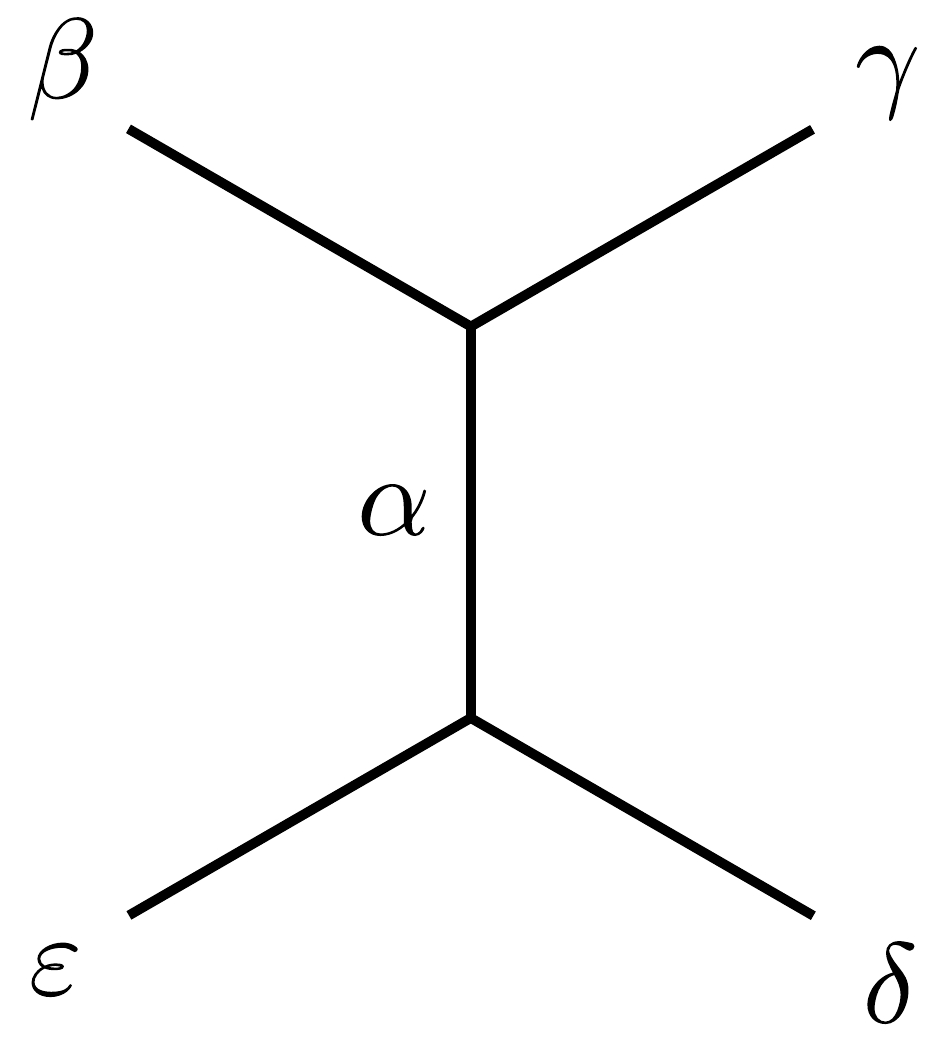}
\centering
\caption{Edge $\alpha$ in the trivalent fat graph $\Gamma$ and adjacent edges.}\label{fig:adj_edges}
\end{figure}

The corresponding Poisson bracket $\{\;,\;\}_{WP}$ 
on $\R^E$ is given by
$$\{f_1,f_2\}_{WP}=(df_1\otimes df_2)(\pi_{WP})\qquad \forall f_1,f_2\in C^\infty(\R^E).$$
In particular, one obtains for the coordinate functions $x^\alpha:\R^E\to \R$ the following expression
\begin{align*}
\{x^\alpha,x^{\alpha'}\}_{WP}=(dx^\alpha\otimes dx^{\alpha'})(\pi_{WP})=\pi_{WP}^{\alpha{\alpha'}} =\delta^{\beta{\alpha'}}-\delta^{\gamma{\alpha'}}+\delta^{\delta{\alpha'}}-\delta^{\epsilon{\alpha'}}.
\end{align*}
The Poisson bracket between the components $c^i$ of the linear constraint  defined in  \eqref{eq:teich_const} and a general  function $f\in C^\infty(\R^E)$ is also easily computed and takes the form
\begin{align}\label{eq:poiss_const}
\{f,c^i\}_{WP}=\!\!\!\sum_{\alpha\in E(\Gamma)}\frac{\partial f}{\partial x^\alpha}\{x^\alpha,c^i\}_{WP}=\!\!\!\sum_{\alpha\in E(\Gamma)}\frac{\partial f}{\partial x^\alpha}\Big(\theta^i{}_\beta-\theta^i{}_\gamma+\theta^i{}_\delta-\theta^i{}_\epsilon\Big).\end{align} 
This vanishes identically for all $f\in C^\infty(\R^E)$, since every face $i\in F(\Gamma)$ involves only left or right turns, as can be seen from the fact that a face containing any of the edges $\alpha,\beta,\gamma,\delta,\epsilon$ in Figure \ref{fig:adj_edges} must involve (a combination of) the edge paths $(\beta,\gamma)$, $(\gamma,\alpha,\delta)$, $(\delta,\epsilon)$ or $(\epsilon,\alpha,\beta)$ or their inverses. 
For all of these paths,  the corresponding sum of multiplicities in \eqref{eq:poiss_const} vanishes and  hence the bracket is trivial.

This demonstrates that the components of the linear constraint  \eqref{eq:teich_const} are Casimir functions for the Poisson bivector \eqref{eq:wp_bivect} and implies  that the Poisson bivector $\pi_{WP}$ restricts to a Poisson bivector for the constraint surface $\Ker c\subset \R^E$. The induced symplectic structure on Teichm\"uller space  coincides with the Weil-Petersson symplectic structure as the following theorem states.
\begin{theorem}[Fock-Chekhov \cite{Fock1999}, Penner \cite{Penner1992,Penner1993}]\label{thm:teich_reduc}

The linear constraint  $c: \R^E\to\R^F$ \eqref{eq:teich_const} is Casimir with respect to the Poisson bivector $\pi_{WP}$ on $\R^E$ and the symplectic quotient determined by $c$  is Poisson isomorphic to Teichm\"uller space with the Weil-Petersson structure.
\end{theorem}

\subsection{The mapping class group action}\label{subsec:back_mcg}
\subsubsection*{Combinatorial description of the mapping class group}
Recall from the previous subsection that mapping class group $\Mod(S)$ is the fundamental group of the Riemann moduli space $\cR(S)$ and is
given as the quotient of the group $\Diff^+(S)$, of orientation preserving diffeomorphisms of $S$, by its normal subgroup $\Diff_0(S)$, of diffeomorphisms isotopic to the identity.
It acts on Teichm\"uller space via pull-back
and a classical result shows this action is properly discontinuous, although not free.

The choice of an embedded trivalent fat graph $\Gamma$ on $S$ gives rise to a combinatorial description of the mapping class group and of its action on Teichm\"uller space in terms of shear coordinates. 
As a first step, we
consider  the action of the mapping class group on embedded fat graphs. Given an element of the mapping class group  $\varphi\in\Mod(S)$ and an embedded fat graph $\Gamma$ one obtains another, combinatorially equivalent, embedded fat graph $\Gamma'=\varphi(\Gamma)$ as the image of $\Gamma$ under $\varphi$.
Clearly, if $\varphi$ is a non-trivial element of $\Mod(S)$, the isotopy classes of $\Gamma$ and $\varphi(\Gamma)$ are necessarily distinct. Conversely, for any two  isotopy classes $\Gamma,\Gamma'$ of embeddings of the same combinatorial trivalent fat graph, there is a unique element of $\varphi\in\Mod(S)$ such that
$\Gamma'=\varphi(\Gamma)$.
Mapping class group elements can therefore be characterised as pairs of isotopy classes of embeddings of a given combinatorial trivalent fat graph. In fact, a result of Penner \cite{Penner1992,Penner1993} allows one to decompose elements of the  mapping class group  into sequences of elementary graph transformations between any such a pair. 

Two embedded trivalent fat graphs $\Gamma$ and $\Gamma'$ are said to be related by a Whitehead move $W_\alpha:\Gamma\mapsto\Gamma'=\Gamma_\alpha$ along an edge $\alpha$ if $\Gamma'$ is obtained from $\Gamma$ by collapsing the edge $\alpha$ into a four-valent vertex and then expanding it in the opposite direction as shown in Figure \ref{fig:whitehead}. 
\begin{figure}[!ht]
\includegraphics[scale=0.3]{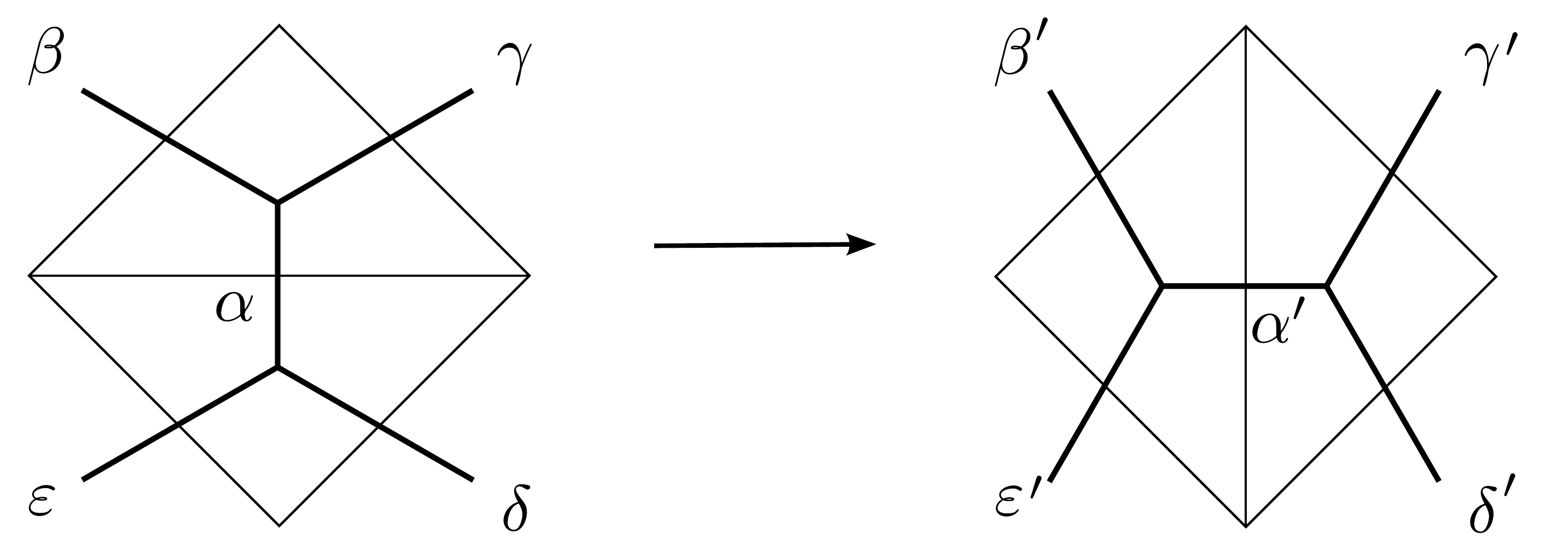}
\centering
\caption{Whitehead move along $\alpha$.}\label{fig:whitehead}
\end{figure}
Similarly, two trivalent fat graphs $\Gamma$ and $\Gamma'$ with ordered edges are said to be related by a transposition of the edges $\alpha$ and $\beta$ if $\Gamma'$ is obtained from $\Gamma$ by exchanging the order of $\alpha$ and $\beta$. Such a transformation will be denoted by $\sigma=(\alpha\;\beta):\Gamma\mapsto\Gamma'$, where $(\alpha\;\beta)\in S_E$ is interpreted as an element of the  symmetric group $S_E$.

The result  in \cite{Penner1992,Penner1993}  states that any  two   embeddings  $\Gamma,\Gamma'$ of a combinatorial edge-ordered trivalent fat graph are related by a sequence of Whitehead moves and transposition of the edges
whose interaction is characterised by a set of simple relations. In particular, this provides a presentation of the mapping class group(oid) in terms of generators and relations.
\begin{theorem}[Penner \cite{Penner1992,Penner1993}]
\label{thm:mcg_fat_graph}
Elements of $\Mod(S)$  are in bijection with finite sequences of elementary graph transformations between embeddings $\Gamma,\Gamma'$  of a combinatorial edge-ordered trivalent fat graph, modulo the following relations:
\begin{enumerate}
 \item (Involutivity) For every edge $\alpha\in E(\Gamma)$ $$W_\alpha^2=\id;$$
 \item (Naturality) For every edge $\alpha$ and every transposition $\sigma$ of edges $$\sigma\circ W_\alpha=W_{\sigma(\alpha)};$$
 \item (Commutativity) For  edges $\alpha,\beta\in E(\Gamma)$which do not share a  common vertex $$W_\alpha\circ W_\beta=W_\beta\circ W_\alpha;$$
 \item (Pentagon) for  edges $\alpha,\beta\in E(\Gamma)$ sharing exactly one vertex $$W_\alpha\circ W_\beta\circ W_\alpha\circ W_\beta\circ W_\alpha=(\alpha\;\beta).$$
\end{enumerate}
\end{theorem}

\subsubsection*{The mapping class group action in shear coordinates}
The description of the mapping class group in terms of elementary graph transformations gives rise to simple and explicit expressions for its action on Teichm\"uller space in terms of shear coordinates. 
Consider again an embedded trivalent fat graph  $\Gamma$ and denote by  $x^\alpha:\cT(S)\to\R$ the coordinate function associated to an edge $\alpha\in E(\Gamma)$. For each element  $\varphi\in\Mod(S)$ denote by  $\Gamma'=\varphi(\Gamma)$ the image of $\Gamma$ under $\varphi$ and  by  $x'{}^{\alpha}:\cT(S)\to\R$  the  coordinate function for the edge  $\varphi(\alpha)\in E(\Gamma')$.
The mapping class group action $\Mod(S)\times\cT(S)\to\cT(S)$ extends to an equivariant action $\Mod(S)\times\R^E\to\R^E$.
This extended action is determined uniquely by the condition that the coordinates of $\varphi^*h\in \cT(S)$ with respect to the embedded graph $\Gamma=\varphi^*(\Gamma')$  agree with the coordinates of $h\in\cT(S)$ with respect to the embedded graph $\Gamma'$. In other words, the coordinate functions $x^\alpha:\cT(S)\to\R$ and $x'{}^{\alpha}:\cT(S)\to\R$ are related by
$$x'{}^{\alpha}=\varphi(x^\alpha)=x^\alpha\circ\varphi^*.$$
From this, it follows that the corresponding group action  $\Mod(S)\times\R^E\to\R^E$ is simply given by the change of coordinates $x^\alpha\mapsto x'{}^{\alpha}$ determined by the two distinct embeddings $\Gamma,\Gamma'$ of the given combinatorial fat graph.

As a consequence of Theorem \ref{thm:mcg_fat_graph} it is then sufficient to consider the transformation of coordinates $x^\alpha\to x'^\alpha$ under Whitehead moves, as their transformation under  edge transpositions is immediate. Choosing a point $h\in\cT(S)$ and comparing its shear coordinates before and after the move,  one obtains an expression for the transformation of the shear coordinates  \cite{Fock1999}.
For the Whitehead move along the edge $\alpha\in E(\Gamma)$ depicted in Figure \ref{fig:whitehead}, the relation between the shear coordinates for $\Gamma$ and those for $\Gamma'=\Gamma_\alpha$ is given by
\begin{align}\label{eq:flip_x}
W_\alpha:\begin{cases}
x^\alpha \mapsto x'{}^{\alpha}=-x^\alpha,
\cr
x^{\beta,\delta}\mapsto x'{}^{\beta,\delta}=x^{\beta,\delta}+\log(1+e^{x^\alpha}),
\cr
x^{\gamma,\epsilon}\mapsto x'{}^{\gamma,\epsilon}=x^{\gamma,\epsilon}-\log(1+e^{-x^\alpha}),
\end{cases}
\end{align}
while all other edge coordinates remain unchanged. 
This formula 
allows one to determine properties of the Whitehead moves  via direct computations. In particular, it follows that they satisfy the conditions in Theorem \ref{thm:mcg_fat_graph}, see \cite{Fock1999}. It is also straightforward  to show that the constraints \eqref{eq:teich_const} and the Poisson bivector \eqref{eq:wp_bivect} are preserved by \eqref{eq:flip_x}. More precisely,  the pull-back of the constraint $c'$, defined with respect to the fat graph $\Gamma'=\Gamma_\alpha$, coincides with the constraint $c$, defined with respect to $\Gamma$, and the push-forward of the bivector $\pi_{WP}$, defined with respect to $\Gamma$, agrees with the bivector $\pi'_{WP}$, defined with respect to $\Gamma'$
$$c=c'\circ W_\alpha,\qquad (W_\alpha)_*\pi_{WP}=\pi_{WP}'.$$
This therefore proves that he Whitehead moves  induce a symplectic mapping class group action on Teichm\"uller space.
\begin{theorem}[Fock-Chekhov \cite{Fock1999}, Penner \cite{Penner1992,Penner1993}]
\label{thm:pentagon_teich}
The Whitehead moves $W_\alpha:\R^E\to\R^E$ \eqref{eq:flip_x} satisfy the relations of Theorem \ref{thm:mcg_fat_graph}. Furthermore, they preserve the constraints $c:\R^E\to\R^F$ \eqref{eq:teich_const} and  the Weil-Petersson Poisson bivector $\pi_{WP}$ \eqref{eq:wp_bivect}.
\end{theorem}

\section{Moduli spaces of 3d gravity}
\label{sec:3dgravity}

In this section we consider moduli spaces of geometric structures which arise  in the context of 3d gravity, namely the moduli spaces of  maximal globally hyperbolic  (MGH) Einstein spacetimes. These moduli spaces are higher-dimensional generalisations of  Teichm\"uller space and classify the diffeomorphism classes of constant curvature Lorentzian metrics on a three-dimensional manifold $M$. In the following,  $M$  denotes a three-dimensional manifold of topology $\R\times S$, where $S$ is a compact orientable genus $g$ surface with $s$ punctures satisfying $2g-2+s>0$. We also restrict  attention to Einstein metrics which are globally hyperbolic with Cauchy surface $S$ and maximal in the sense that any isometric embedding of $M$ into another globally hyperbolic spacetime $N$ is a global isometry. See \cite{Beem1996} for details on causality of Lorentzian manifolds.

\subsection{The isometry groups in 3d gravity}
 \label{subsec:isom}
\subsubsection*{The isometry groups and their Lie algebras} 
The simplicity of 3d gravity is a consequence of the vanishing of the traceless part of the Riemann tensor, the Weyl tensor, in three dimensions. It implies that the Riemann curvature  tensor and therefore the sectional curvature of a three-dimensional manifold are determined uniquely by its Ricci tensor. It then follows  that any  Einstein spacetime, a solution of Einstein equations with vanishing stress-energy tensor,  is  locally isometric to one of three model Lorentzian manifolds with sectional curvature given by the cosmological constant $\Lambda$.

These model spacetimes are  three-dimensional  Minkowski space  $\mathrm{M}_3$ for $\Lambda=0$,  anti-de Sitter space $\mathrm{AdS}_3$ for $\Lambda=-1$, and  de Sitter space $\mathrm{dS}_3$ for $\Lambda=1$. In the following we  denote these model spacetimes by $X_\Lambda$, their isometry groups by $G_\Lambda=\Isom(X_\Lambda)$ and the associated Lie algebras by $\mathfrak g_\Lambda=\mathrm{Lie}(G_\Lambda)$. As solutions for different values of the curvature $\Lambda$ can be obtained by simple rescalings of the metrics, 
 we restrict attention to the cases $\Lambda=0,-1,1$.

The three model spacetimes have a simple description in terms of the group $\mathrm{PSL}(2,\R)$, which is outlined in Appendix \ref{subsec:model}, and their isometry groups are given by
\begin{align*}
G_\Lambda=\begin{cases}
\mathrm{PSL}(2,\R)\ltimes \mathfrak{sl}(2,\R) & \Lambda=0\\
\mathrm{PSL}(2,\R)\times \mathrm{PSL}(2,\R) & \Lambda=-1\\
\mathrm{PSL}(2,\C) &\Lambda=1.
\end{cases}
\end{align*}
In all cases, the associated  Lie algebra $\mathfrak g_\Lambda$ is a six-dimensional real Lie algebra  and can be described in terms of a common basis for which the cosmological constant plays the role of a structure constant \cite{Witten1988/89}. This basis  involves 
a basis $\{J_i\}_{i=0,1,2}$ 
 of $\mathfrak{sl}(2,\R)$ and three additional basis vectors $\{P_i\}_{i=0,1,2}$ such that Lie bracket is given by
\begin{align}\label{liebr}
[J_i,J_j]=\sum_{k=0}^2 \epsilon_{ij}{}^{k} J_k,\quad[J_i,P_j]=\sum_{k=0}^2\epsilon_{ij}{}^{k}P_k,\quad [P_i,P_j]=-\Lambda\sum_{k=0}^2\epsilon_{ij}{}^{k}J_k,
\end{align}
where $\epsilon_{ijk}$ is the totally antisymmetric tensor in three dimensions with $\epsilon_{012}=1$ and indices are raised and lowered with the three-dimensional Minkowski metric $\eta=\text{diag}(-1,1,1)$.
For $\Lambda=0$, this is simply the  Poincar\'e algebra in three dimensions. For $\Lambda=1$ and $\Lambda=-1$, one can introduce the alternative basis
$\{J_i^\pm\}_{i=0,1,2}$ with $J^\pm_i =\tfrac 1 2 \left(J_i\pm 
{P_i}/\sqrt{-\Lambda} \right)$, 
in which the Lie bracket of $\mathfrak g_\Lambda$ reads
$$
[J_i^\pm, J_j^\pm]=\epsilon_{ij}{}^kJ^\pm_k,\qquad[J_i^\pm, J_j^\mp]=0.
$$
This shows that the Lie algebra \eqref{liebr} is isomorphic to  $\mathfrak{sl}(2,\R)\ltimes\mathfrak{sl}(2,\R)$ for $\Lambda=0$, to  $\mathfrak{sl}(2,\R)\oplus \mathfrak{sl}(2,\R)$ for $\Lambda=-1$ and to $\mathfrak{sl}(2,\C)$ for $\Lambda=1$.
In the following, we also need to consider $\Ad$-invariant symmetric bilinear forms on $\mathfrak{g}_\Lambda$. In all three cases the real vector space of such bilinear forms is two-dimensional, with a basis given by the forms $(\;,\;),\langle\;,\;\rangle:\mathfrak g_\Lambda\times\mathfrak g_\Lambda\to \R$ defined by
\begin{align}\label{eq:forms}
&(J_i,P_j)=0, & & (J_i, J_j)=\eta_{ij}, &  & (P_i,P_j)=-\Lambda\eta_{ij},
\\
&\langle J_i, P_j\rangle=\eta_{ij}, &  & \langle J_i,J_j\rangle=0, & & \langle P_i,P_j\rangle=0.\nonumber
\end{align}

\subsubsection*{A unified description of the Lie algebras} 
A convenient description of  the isometry groups $G_\Lambda$ and their Lie algebras $\mathfrak g_\Lambda$ is obtained by exteding $\mathfrak{sl}(2,\R)$ to a Lie algebra over a commutative real algebra
$R_\Lambda$, see \cite{Meusburger2007}. As a vector space, this algebra $R_\Lambda$ is isomorphic to $\R^2$ and its multiplication law  is given by
\begin{align*}
(x,y)\cdot (u,v)=(xu-\Lambda yv, xv+yu)\qquad \forall x,y,u,v\in\R.
\end{align*}
Writing  $1=(1,0)$ and $\ell=(0,1)$, one obtains a parametrisation of  $R_\Lambda$  analogous to the complex numbers,
and consequently we use the notation  $\mathrm{Re}_\ell(x+\ell y)=x$, $\mathrm{Im}_\ell(x+\ell y)=y$ for all $x,y\in\R$.
A direct computation shows that $R_\Lambda$ is isomorphic to $\C$  with $\ell=\mathrm{i}$ for $\Lambda=1$, to the split-complex numbers for $\Lambda=-1$, and to the dual numbers for $\Lambda=0$.

Note that, for $\Lambda=0,-1$,  the algebra $R_\Lambda$ has zero divisors. For $\Lambda=-1$ these are of the form $\tfrac 1 2 (1\pm \ell) x$  with $x\in \R$, and one has
\begin{align*}
\tfrac 1 2 (1\pm \ell)\cdot \tfrac 1 2 (1\pm \ell)=\tfrac 1 2 (1\pm \ell),\qquad \tfrac 1 2 (1\pm \ell)\cdot \tfrac 1 2 (1\mp\ell)=0.
\end{align*}
This allows one to extend analytic functions $f:\R\to\R$ to analytic functions $f:R_{-1}\to R_{-1}$ 
\begin{align*}
f(x+\ell y)&
=\tfrac 1 2(1+\ell) f(x+y)+\tfrac 1 2 (1-\ell)f(x-y).
\end{align*}
For $\Lambda=0$, the zero divisors in $R_\Lambda$ are of the form $\ell y$ with $y\in\R$.  Analytic functions $f:\R\to\R$ can thus be extended to functions $f:R_0\to R_0$ via
\begin{align*}
f(x+\ell y)=f(x)+\ell f'(x)y.
\end{align*}
Note that these expressions generalise the extension of real analytic functions $f:\R\to\R$ to complex analytic  functions $f:\C\to \C$ and also give rise to the following generalisation of the Cauchy-Riemann differential equations 
\begin{align*}
\frac{\partial \mathrm{Re}_\ell f}{\partial x}=\frac{\partial \mathrm{Im}_\ell f}{\partial y},\qquad \frac{\partial \mathrm{Re}_\ell f}{\partial y}=-\Lambda\,\frac{\partial \mathrm{Im}_\ell f}{\partial x}.
\end{align*}
The algebra $R_\Lambda$  allows one to identify  the Lie algebras $\mathfrak g_\Lambda$ with the Lie algebra of traceless $2\times 2$-matrices with entries in $R_\Lambda$ \cite{Meusburger2007}. By considering a basis $\{J_i\}_{i=0,1,2}$ of $\mathfrak{sl}(2,\R)$ with Lie bracket $[J_i,J_j]=\epsilon_{ij}{}^kJ_k$ and setting $P_i=\ell J_i$ one  obtains the Lie algebras $\mathfrak g_\Lambda$ with Lie bracket \eqref{liebr}. This description of the Lie algebras $\mathfrak{g}_\Lambda$ in terms of $R_\Lambda$ also gives rise to a common description of the bilinear forms \eqref{eq:forms}. They are obtained as the real and imaginary part of the bilinear extension of the Killing form $\kappa$ on $\mathfrak{sl}(2,\R)$ to $\mathfrak g_\Lambda$
\begin{align}\label{formreal}
(\;,\;)=2\mathrm{Re}_\ell\;\kappa\qquad \langle\;,\;\rangle=2\mathrm{Im}_\ell\; \kappa,
\end{align}
where $\kappa(X,Y)=\mathrm{Tr}(XY)$ for $X,Y\in\mathfrak g_\Lambda$.

\subsection{Einstein spacetimes and their moduli spaces}
\label{subsec:ghm_spt}
\subsubsection*{Classification of Einstein spacetimes}
Generalising the results summarised in Section \ref{sec:background}, we now consider the moduli spaces of maximal globally hyperbolic (MGH) Einstein metrics of curvature $\Lambda$ on $M$ that induce a complete metric of  finite area on the Cauchy surface $S$  modulo orientation preserving diffeomorphisms
$$\cM_\Lambda(M)=\Ein_\Lambda(M)/\Diff^+(M).$$
As in the case of the Riemann moduli space, it is convenient to consider the universal covering space of $\cM_\Lambda(S)$ by identifying only those metrics that are related by the subgroup of diffeomorphisms isotopic to the identity. This leads to  the Teichm\"uller-like moduli spaces
$$\cG\cH_\Lambda(M)=\Ein_\Lambda(M)/\Diff_0(M).$$

One approach to the classification of MGH Einstein spacetimes in three dimensions
is based on their
description as quotients of regions in the model spacetimes by a discrete group of isometries \cite{Witten1988/89,Mess2007,Scannell1999,Benedetti2009}. 
The resulting classification  is analogous to  the uniformisation theorem for hyperbolic surfaces in two dimensions and states 
 that MGH Einstein spacetimes are largely determined by their holonomy representation $\rho:\pi_1(M)\cong\pi_1(S)\to G_\Lambda$, which defines the action of $\pi_1(S)$ on the universal cover of $M$.
More precisely, for $\Lambda=0,-1$, a MGH Einstein metric $g\in \cG\cH_\Lambda(M)$ can be described as follows. First, consider the universal cover  $\tilde S$ of the Cauchy surface $S$. It is shown in \cite{Mess2007} that $\tilde S$ isometrically embeds in $X_\Lambda$ and that the universal cover $\tilde M$ of $M$ is obtained from this embedding. For $\Lambda=-1$ $\tilde M$ coincides with the domain of dependence of $\tilde S$ and for $\Lambda=0$ it is the chronological future of this domain of dependence. In other words, the universal cover of $M$ isometrically embeds in $X_\Lambda$ and the group of deck transformations provides a representation of $\pi_1(M)\cong\pi_1(S)$ into $G_\Lambda$.

Due to the non-degeneracy of the three-dimensional metric, not all representations $\rho:\pi_1(S)\to G_\Lambda$ arise 
as holonomy representations of MGH spacetimes. The allowed representations are also described in \cite{Mess2007} and can be characterised as follows.  For $\Lambda=-1$ the allowed holonomy representations $\rho:\pi_1(S)\to\mathrm{PSL}(2,\R)\times\mathrm{PSL}(2,\R)$ are the ones that decompose into two Fuchsian components $\rho_{l,r}=\pr_{1,2}\circ \rho:\pi_1(S)\to\mathrm{PSL}(2,\R)$. For $\Lambda=0$, they are representations $\rho:\pi_1(S)\to\mathrm{PSL}(2,\R)\ltimes\mathfrak{sl}(2,\R)$ which decompose into a Fuchsian part $\rho_0=\pr_1\circ \rho:\pi_1(S)\to\mathrm{PSL}(2,\R)$ and a $\rho_0$-cocycle $\tau=\pr_2\circ H:\pi_1(S)\to\mathfrak{sl}(2,\R)$ with
$$\tau(ab)=\tau(a)+\rho_0(a)\tau(b)\rho_0(a)^{-1},\qquad\forall a,b\in \pi_1(S).$$
The moduli spaces for $\Lambda=-1,0$ are thus shown to be in one-to-one correspondence with certain components of the representation variety $\Rep_\cC(S,\mathrm{G}_\Lambda)$.

For $\Lambda=1$ the correspondence between holonomy representations and MGH spacetimes  is only locally injective, which means that the holonomy data is not sufficient to distinguish certain MGH de Sitter spacetimes. This is a consequence of the fact that the universal cover of $M$ is  in general only immersed in $X_1$. The suitable data for the classification of such spacetimes is obtained by grafting of hyperbolic surfaces along measured geodesic laminations, which also provides an alternative description for the flat and AdS case.

\subsubsection*{Grafting parametrisation}

As explained in Appendix \ref{subsec:model}, all three model spacetimes $X_\Lambda$ are equipped with certain embeddings of the hyperbolic plane, either in $X_\Lambda$ itself ($\Lambda=-1,0$) or in an appropriate dual space ($\Lambda=1$). For a spacetime with purely Fuchsian holonomy $\rho_0:\pi(S)\to\mathrm{PSL}(2,\R)\subset G_\Lambda$, the action of the holonomy group on $X_\Lambda$ induces an action on these embedded hyperbolic planes and thus defines a hyperbolic surface $h\in\cT(S)$. Applying earthquakes along measured laminations for $h$, one then obtains a description of all such Fuchsian spacetimes in terms of the fibre $\cM\cL_h(S)$.

General MGH spacetimes, whose holonomies are not restricted to the subgroup $\mathrm{PSL}(2,\R)\subset G_\Lambda$, are obtained as deformations of 
these  Fuchsian spacetimes via grafting. Grafting is an operation $Gr:\cM\cL(S)\to\cG\cH_\Lambda(M)$ that associates to each point $h\in\cT(S)$ and each measured geodesic lamination $\lambda\in\cM\cL_h(S)$ a MGH metric $Gr^\lambda(h)\in\cG\cH_\Lambda(M)$, called the grafting of $h$ along $\lambda$. In terms of representations, the construction is similar to the one of earthquakes on Riemann surfaces and can be described as follows. 

Consider a point  $h\in \cT(S)$ with associated Fuchsian representation $\rho_0:\pi_1(S)\to\mathrm{PSL}(2,\R)$. Then the grafting of $h$ along $\lambda\in\cM\cL_h(S)$   corresponds to another representation $\rho:\pi_1(S)\to G_\Lambda$ that is the product $\rho=Z_G^\lambda\cdot\rho_0$ of the Fuchsian representation $\rho_0$  and a grafting cocycle $Z_G^\lambda:\pi_1(S)\to G_\Lambda$. Such cocycles are defined in exactly the same way as the cocycles for earthquakes in \eqref{eq:zedef}, only now all measures are multiplied by the $R_\Lambda$-imaginary unit $\ell$. Let $\tilde\lambda\in\cM\cL(\bbH^2)$ be the lift of $\lambda$ to the universal cover of $(S,h)$. As the hyperbolic plane is embedded in $X_\Lambda$ (or in the associated dual space), there is a well defined notion of rotations around any leaf $\tilde l$ of $\tilde\lambda$. For any point $p\in \tilde l$, this rotation is given by $\Ad_{A_p}E(\ell \mu_p)$, where $A_p$ is the hyperbolic isometry mapping the imaginary axis to $\tilde l$,  $\mathrm{i}
$ to $p$ and preserving the orientation, $\mu_p$ is the associated oriented weight, as in \eqref{eq:zedef}, which determines  the angle of rotation and
\begin{align}
E(\ell\mu)=\begin{cases}(\id,\mu J_1) & \Lambda=0, \cr (E(\mu),E(-\mu)) & \Lambda=-1, \cr  E(\mathrm{i}\mu) & \Lambda=1.\end{cases}
\end{align}
Here $J_1$ denotes the generator of hyperbolic translations along the imaginary axis on $\bbH^2$, see \eqref{eq:jgen}.
Note that this rotation is a direct generalisation of the hyperbolic translation along a geodesic in \eqref{eq:coord_earth} and the cocycle 
 $Z_G^\lambda:\pi_1(S)\to G_\Lambda$ is obtained as a direct generalisation of the cocycle \eqref{eq:zedef}
\begin{align}\label{eq:gr_cocyle}
Z_G^\lambda(a)=Z_E^{\ell\lambda}(a)
=\prod_{p\in \lambda\cap a}\Ad_{A_p}E(\ell\mu_p).
\end{align}
In particular, this expression makes it clear that for $\Lambda=1$ the correspondence between measured laminations and de Sitter holonomy representations can only be locally injective, since laminations with the same support whose measures differ by multiples of $2\pi$ give rise to the same cocycle and hence to the same holonomy representation.
This correspondence between MGH Einstein spacetimes and measured laminations allows one to identify the former with the bundle of measured geodesic laminations over Teichm\"uller space. 
\begin{theorem}[Mess \cite{Mess2007}, Scannell \cite{Scannell1999}, Benedetti-Bonsante \cite{Benedetti2009}]
 Let $S$ be a closed orientable surface of  genus $g$ and  with $s$ punctures satisfying $2g-2+s>0$. Then the Teichm\"uller-like moduli spaces of MGH Einstein spacetimes on $M=\R\times S$ are  homeomorphic to the bundle of measured geodesic laminations over Teichm\"uller space
$$\cG\cH_\Lambda(M)\cong\cM\cL(S).$$
\end{theorem}

\subsubsection*{The gravitational symplectic structure}
\label{subsec:symp}
From a physics viewpoint, the parametrisation of MGH spacetimes in terms of holonomies is closely related to the Chern-Simons formulation of 3d gravity developed in \cite{Achucarro1986,Witten1988/89}. In this formulation, the spacetime metric is first decomposed into a (co-)frame field $e$ and an associated spin connection $\omega$, which are then combined into a $G_\Lambda$-connection
$$A=\omega^iJ_i+e^iP_i=(\omega^i+\ell e^i)J_i,$$
where $J_i$ and $P_i$ denote the basis of $\mathfrak{g}_\Lambda$ introduced in \eqref{liebr}. The requirements of flatness and vanishing torsion on $e$ and $\omega$  translate into a flatness condition $F=dA+A\wedge A=0$ for  the $G_\Lambda$-connection.   This allows one to relate the moduli spaces of MGH Einstein spacetimes of curvature $\Lambda$ to the moduli space of flat $G_\Lambda$ on the Cauchy surface $S$ \cite{Witten1988/89}. In particular, the gravitational symplectic structure on $\cG\cH_\Lambda(M)$ can be characterised in terms of the Chern-Simons symplectic structure, which agrees with the Atiyah-Bott symplectic structure on $\Rep_\cC(S,G_\Lambda)$ \cite{Atiyah1983}.
 
We start by summarising the relevant results on this symplectic structure for a general structure group $G$. Given an $\Ad$-invariant, non-degenerate  symmetric bilinear form $B:\mathfrak g\times\mathfrak g\to\R$ on the Lie algebra $\mathfrak g=\mathrm{Lie}\, G$, one obtains a canonical symplectic structure 
 on the  moduli space $\Rep_\cC(S,G)$ \cite{Atiyah1983,Goldman1984}. It was shown by Goldman \cite{Goldman1986}
 that the corresponding Poisson structure can be expressed in terms of the bilinear form $B$ and  the intersection behaviour of curves representing  elements of $\pi_1(S)$ as follows.
To each class function $f\in C_G(G)$ and  each element $a\in\pi_1(S)$ one associates a function $f_a:\Rep_\cC(S,G)\to\R$ defined by $f_a(\rho)=f(\rho(a))$. Then, the Goldman Poisson bracket between two such functions  $f_a,g_b$ is defined  by 
\begin{align}\label{eq:goldman}
\{f_a,g_b\}_G(\rho)=\sum_{p\in a\cap b} \epsilon_p(a,b)\, B(F_{a_p}(\rho), G_{b_p}(\rho)),
\end{align}
where the sum is over the intersection points of $a,b\in\pi_1(S)$ and  $\epsilon_p(a,b)$ denotes their oriented intersection number at $p$. The indices $a_p,b_p$ on the right-hand side of \eqref{eq:goldman} stand for representatives of $a$ and $b$ based at the point $p$, and the functions  $F_{a_p}, G_{b_p}:\Rep_\cC(S,G)\to\mathfrak{g}$ are defined by
\begin{align*}
B(F_{a_p}(\rho),X)=\frac d {dt}\bigg\vert_{t=0} f(\rho(a_p)\,e^{tX}),
\qquad \forall X\in\mathfrak{g}.
\end{align*}

In the following, it will be convenient to express this Poisson structure \eqref{eq:goldman} in terms of a  basis $\{T_k\}$ of $\mathfrak g$. 
Denoting by $B_{kl}=B(T_k,T_l)$ the coefficients  of the Ad-invariant symmetric bilinear form $B$ with respect to this basis
 and by $B^{kl}$  the entries of the inverse of the coefficient matrix,  one find that the bracket \eqref{eq:goldman}  is given by
\begin{align}
\label{eq:b_gold}
B(F_{a_p}(\rho), G_{b_p}(\rho))=\sum_{k,l=1}^{\dim(\mathfrak{g})}  B^{kl}B(F_{a_p},T_k)B(G_{b_p},T_l).
\end{align}

Note that although the definition of this bracket involves a choice of paths on $S$ which represent the elements $a,b\in \pi_1(S)$, it is shown in \cite{Goldman1986} that the bracket  does not depend on this choice and induces a symplectic structure on $\Rep_\cC(S,G)$.

\begin{theorem} [Goldman \cite{Goldman1984,Goldman1986}]
Formula \eqref{eq:goldman} defines a symplectic structure 
 on $\Rep_\cC(S,G)$ which coincides with  the Atiyah-Bott symplectic structure.
\end{theorem}

In particular, Goldman's symplectic structure can be used to describe the Weil-Petersson structure on Teichm\"uller space, by realising the latter as a connected component 
of $\Rep_\cC(S,\mathrm{PSL}(2,\R))$. In other words, for the structure group $\mathrm{PSL}(2,\R)$ and the the Killing form $\kappa$ on $\mathfrak{sl}(2,\R)$,  the restriction of the Goldman Poisson bracket \eqref{eq:goldman} to the Teichm\"uller component of $\Rep_\cC(S,\mathrm{PSL}(2,\R))$ induces the Weil-Petersson symplectic structure on Teichm\"uller space.

Similarly, the symplectic structures on the moduli spaces of 3d MGH spacetimes are closely related to the Goldman bracket on $\Rep_{\mathcal C}(S,G_\Lambda)$. However, unlike in the case of $\mathrm{PSL}(2,\R)$, the space of $\Ad$-invariant symmetric bilinear forms on $\mathfrak g_\Lambda$ is two-dimensional, and there are inequivalent versions of the Goldman bracket on these moduli spaces, corresponding to different  linear combinations of the bilinear forms $(\;,\; )$ and $\langle\;,\;\rangle$ defined in \eqref{eq:forms}. It is shown in \cite{Witten1988/89} that the bilinear forms on $\mathfrak g_\Lambda$ that are relevant for 3d gravity 
are the forms $\langle\;,\;\rangle$, which according to  equation \eqref{formreal} can be interpreted as  the imaginary part of the bilinear extension of the Killing form $\kappa$ to $\mathfrak g_\Lambda$.
\begin{theorem} [Witten \cite{Witten1988/89}]\label{thm:witten}
The gravitational Poisson structure on the Teichm\"uller-like moduli spaces of MGH spacetimes agrees with the restriction of the imaginary part of the Atiyah-Bott-Goldman Poisson structure on the $G_\Lambda$-representation variety.
\end{theorem}
Only for this choice of the bilinear form, the Chern-Simons action agrees with the Einstein-Hilbert action for 3d gravity in Cartan's formulation. Other choices of the $\Ad$-invariant symmetric bilinear form on $\mathfrak g_\Lambda$ yield a different action which gives rise to  the same equations of motion but  induces a different  symplectic structure on the moduli space. It should therefore be expected that the choice of the correct bilinear form has important consequences for the resulting quantum theory.

\section{Generalised shear coordinates}\label{sec:general_shear}

We are now ready to introduce generalised shear-bending coordinates on the Teichm\"uller-like moduli spaces $\cG\cH_\Lambda(M)$. This will be achieved by parametrising the deformation cocycles in terms of analytic shear coordinates on $\cM\cL^{R_\Lambda}(S)$, the bundle of $R_\Lambda$-valued measured geodesic laminations. Using Thurston's earthquake theorem, we first define shear coordinates on the bundle $\cM\cL(S)$ of measured geodesic laminations and, using the fact that the earthquake cocycles depend  analytically on the set of measures on a lamination with fixed support, we then define an analytic extension of these coordinates to $\cM\cL^{R_\Lambda}(S)$. This construction gives rise to coordinates on the moduli spaces $\cG\cH_\Lambda(M)$ that have a clear  geometric interpretation in terms of grafting along ideal edges of an ideal triangulation of $S$. We then derive an expression for the gravitational symplectic structure on $\cG\cH_\Lambda(M)$ in terms of these  coordinates and describe 
its relation to the Weil-Petersson symplectic structure  and to the cotangent bundle over Teichm\"uller space.

\subsection{Definition of shear-bending coordinates}
\label{subsec:general_shear}

\subsubsection*{Shear coordinates for $\cM\cL(S)$ via Thurston's theorem}
To construct generalised shear coordinates, we first show how Thurston's shear coordinates on Teichm\"uller space 
provide  a global parametrisation of the bundle $\cM\cL(S)$ of measured geodesic laminations via Thurston's earthquake theorem, see end of Subsection \ref{subsec:back_teichmueller}.

Let $h\in\cT(S)$ be a point in Teichm\"uller space, $\lambda\in\cM\cL_h(S)$ a measured geodesic lamination for $h$ and $Eq^\lambda(h)\in\cT(S)$ the earthquake  of $h$ along $\lambda$. For an embedded trivalent fat graph $\Gamma$, denote by $x^\alpha=x^\alpha(h)$ and $x'{}^\alpha=x^\alpha(Eq^\lambda(h))$ the shear coordinates of $h$ and $Eq^\lambda(h)$ associated to an edge $\alpha\in E(\Gamma)$. Comparing the holonomy representations  \eqref{eq:hol} of $h$ and $Eq^\lambda(h)$ in terms of shear coordinates, we obtain the following parametrisation of the associated earthquake cocycle $Z^\lambda_E:\pi_1(S)\to\mathrm{PSL}(2,\R)$ defined in \eqref{eq:zedef}
\begin{align}\label{eq:cocycle}
&Z_E^\lambda(a)=\rho^\lambda(a)\rho(a)^\inv
=\Ad_{A^a_n(x)} E(x'{}^{\alpha_n}\!\!-\!x{}^{\alpha_n})\cdots 
\Ad_{A^a_1(x)}E(x'{}^{\alpha_1}\!\!-\!x{}^{\alpha_1}).
\end{align}
Here $(\alpha_1,...,\alpha_n)$ is the sequence of edges of $\Gamma$ representing $a\in\pi_1(S)$ and $A^a_k(x)$ is the hyperbolic isometry that maps the imaginary axis on $\bbH^2$ to the lift of the ideal geodesic dual to $\alpha_k$
\begin{align}\label{eq:adef}
A^a_k(x)
=P^a_n E(x^{\alpha_n})P^a_{n-1}\cdots E(x^{\alpha_{k+1}}) P^a_k.
\end{align}
This expression for the earthquake cocycle $Z_E^\lambda$ in terms of the difference between the shear coordinates of $Eq^\lambda(h)$ and $h$ then allows us to define coordinates $u^\alpha:\cM\cL_h(S)\to\R$ parametrising the fibres of $\cM\cL(S)$ via
\begin{align}\label{eq:udef}
u^\alpha(\lambda)=x^\alpha(Eq^\lambda(h))-x^\alpha(h),\qquad \forall \lambda\in\cM\cL_h(S).
\end{align}
Clearly, these coordinates are not all independent but satisfy the same constraints as the shear coordinates $x^\alpha$, namely for each face $i\in F(\Gamma)$
\begin{align}\label{eq:u_const}
c^i(u)=\sum_{\alpha\in E(\Gamma)}\theta^i{}_\alpha u^\alpha=0.
\end{align}
Interpreting the constraints for the different faces $i\in F(\Gamma)$ as components of a linear map $c:\R^E\to\R^F$,  we thus obtain an explicit description of measured geodesic lamination in terms of shear coordinates on $S$, which characterises $\cM\cL(S)$ as a linear subspace of $\R^E\times\R^E$.
\begin{proposition}\label{lem:embed}
The coordinate functions $x^\alpha,u^\alpha:\cM\cL(S)\to\R$ define an embedding $(x,u):\cM\cL(S)\hookrightarrow \R^E\times \R^E$ whose image is the kernel of the linear map $c\oplus c:\R^E\times \R^E \to\R^F\times \R^F$ whose components are given by \eqref{eq:teich_const} and \eqref{eq:u_const}.
\end{proposition}

Note also that the coordinate functions $u^\alpha$ satisfy certain cocycle conditions reminiscent from the properties of the cocycle \eqref{eq:cocycle}. For two measured geodesic laminations $\lambda_1,\lambda_2\in\cM\cL_h(S)$ denote by $h_1=Eq^{\lambda_1}(h)$ and $h_2=Eq^{\lambda_2}(h)$ the images of $h$ under the associated earthquakes and let $\lambda'\in\cM\cL_{h_1}(S)$ be the measured geodesic lamination with $h_2=Eq^{\lambda'}(h_1)$. Then the definition of $u^\alpha$
directly implies  the cocycle condition
$$u^\alpha(\lambda')=u^\alpha(\lambda_2)-u^\alpha(\lambda_1).$$

\subsubsection*{Analytic extension to $\cM\cL^{R_\Lambda}(S)$}
An analogous  description of grafting construction is obtained by considering geodesic laminations with $\ell$-imaginary measures and analytic continuation of the coordinates $u^\alpha$ defined above. In fact, we may consider more general $R_\Lambda$-valued measured laminations, defined as pairs $(\lambda,\mu+\ell\nu)$ with $\mu$ and $\nu$ real transverse measures supported on $\lambda$. Given a point $h\in\cT(S)$ and an $R_\Lambda$-valued measured lamination $\lambda\in \cM\cL_h^{R_\Lambda}(S)$, we  define the $\ell$-complexified earthquake of $h$ along $\lambda$ in terms of the cocycle $Z_{EG}^\lambda:\pi_1(S)\to G_\Lambda$ 
$$Z_{EG}^\lambda(a)=\prod_{p\in \lambda\cap a}\Ad_{A_p}E(\mu_p+\ell\nu_p),\qquad \forall a\in\pi_1(S),$$
with the same notation as in \eqref{eq:zedef} and
\begin{align}\label{eq:e_gen}
E(\mu+\ell\nu)=\begin{cases}(E(\mu),\nu J_1) & \Lambda=0, \cr (E(\mu+\nu),E(\mu-\nu)) & \Lambda=-1, \cr  E(\mu+\mathrm{i}\nu) & \Lambda=1.\end{cases}\end{align}
Note that earthquake and grafting cocycles in \eqref{eq:zedef} and \eqref{eq:gr_cocyle} are obtained as particular cases of this $R_\Lambda$-valued cocycle for purely real or purely imaginary $R_\Lambda$-measures.

For an embedded trivalent fat graph $\Gamma$ on $S$, denote by $u^\alpha:\cM\cL(S)\to\R$  the shear coordinates for $\R$-valued measured laminations, as defined in \eqref{eq:udef}. We now wish to analytically extend the coordinate functions $u^\alpha$ to $\cM\cL^{R_\Lambda}(S)$. This is indeed possible since the earthquake map $Eq:\cM\cL(S)\to\cT(S)$ depends analytically on the measure of laminations with fixed support \cite{McMullen1998} and the shear coordinates $x^\alpha:\cT(S)\to\R$ on Teichm\"uller space are also analytic \cite{Penner2012}. Together these imply the following result.
\begin{proposition}
The coordinates $u^\alpha:\cM\cL(S)\to\R$ are analytic on the measure of laminations with fixed support and therefore admit a unique analytic extension $w^\alpha:\cM\cL^{R_\Lambda}(S)\to R_\Lambda$ satisfying
$w^\alpha|_{\cM\cL(S)}=u^\alpha$.
\end{proposition}
This allows to describe the associated shear-bending cocycle $Z_{EG}^\lambda:\pi_1(S)\to G_\Lambda$ via
\begin{align*}
Z_{EG}^\lambda(a)=\Ad_{A_{n}^a(x)}E(u^{\alpha_{n}}+\ell v^{\alpha_{n}})\cdots \Ad_{A_{1}^a(x)}E(u^{\alpha_{1}}+\ell v^{\alpha_{1}}),
\end{align*}
with the same notation as in \eqref{eq:cocycle}, $u^\alpha=\mathrm{Re}_\ell(w^\alpha)$, $v^\alpha=\mathrm{Im}_\ell(w^\alpha)$ and $E(u+\ell v)$ as in \eqref{eq:e_gen}.

\subsubsection*{Shear-bending coordinates on $\cG\cH_\Lambda(M)$}

The analytic extension $w^\alpha$ of the coordinates $u^\alpha$ in \eqref{eq:udef} now allows us to  define generalised shear coordinates on $\cG\cH_\Lambda(M)$ as follows. Consider a point $g\in\cG\cH_\Lambda(M)$,  let  $h\in\cT(S)$ be  the hyperbolic metric determined by the Fuchsian part of its holonomy representation and  $\lambda\in\cM\cL_h(S)$ the measured lamination associated with its grafting cocycle. Denoting 
 by $x^\alpha(h)$ the shear coordinate of $h$ and by $w^\alpha(\ell\lambda)$ the shear coordinate of $\ell\lambda$ for each edge $\alpha\in E(\Gamma)$, we define 
the generalised shear coordinate  of $g$ as
\begin{align}\label{eq:zdef}
z^\alpha(g)=x^\alpha(h)+w^\alpha(\ell\lambda)=x^\alpha(g)+\ell y^\alpha(g).
\end{align}
Note that $\Re_\ell(w^\alpha(\ell\lambda))$ may in general be non-zero and, consequently,   the real part $x^\alpha(g)$ of the generalised shear coordinates does not necessarily agree with the corresponding shear coordinates $x^\alpha(h)$.

As a direct consequence of the definition of generalised shear coordinates, one finds that the holonomy representation of $g$ can be parametrised exactly as in \eqref{eq:hol} by
\begin{align}\label{eq:hol_shear}
\rho(a)=P^a_n E(z^{\alpha_n})\cdots P^a_1 E(z^{\alpha_1}),
\end{align}
where  the real matrices $E(x^{\alpha_k})$ in \eqref{eq:hol} are replaced by $R^\Lambda$-valued matrices $E(z^{\alpha_k})$.
These terms can  be interpreted  as a combination of earthquakes and grafting along the ideal edges of the  triangulation dual to $\Gamma$. We therefore refer to the generalised shear coordinates $z^\alpha$ as shear-bending coordinates. Note that these coordinates can also be viewed as generalisations of the shear-bending coordinates of  Bonahon \cite{Bonahon1996} to the Lorentzian context.

As in the real case, the parametrisation \eqref{eq:hol_shear} of the holonomies in terms of shear-bending coordinates gives rise to  constraints associated with the faces of $\Gamma$. A direct computation  yields
\begin{align}\label{eq:const_z}
c^i_\Lambda(z)=\sum_{\alpha\in E(\Gamma)} \theta^i{}_\alpha z^\alpha=0,
\end{align}
for each face  $i\in F(\Gamma)$. These constraints again impose the tracelessness of the holonomies around the punctures and allow one to realise the moduli spaces $\cG\cH_\Lambda(M)$ of 3d spacetimes as a linear subspaces of $R^E_\Lambda$, thus generalising Theorem \ref{thm:teich_shear}.
\begin{theorem}\label{thm:gen_embed}
The coordinate functions $z^\alpha:\cG\cH_\Lambda(S)\to R_\Lambda$ in \eqref{eq:zdef} define an embedding $z:\cG\cH_\Lambda(S)\hookrightarrow R^E_\Lambda$ whose image agrees with the kernel of the linear map $c_\Lambda:R^E_\Lambda \to R^F_\Lambda$ whose components are given by \eqref{eq:const_z}.
\end{theorem}

\subsection{The gravitational symplectic structures}
\label{subsec:symp_gen}
\subsubsection*{Symplectic structure in terms of shear-bending coordinates}
We now describe how the gravitational symplectic structure on the moduli space $\cG\cH_\Lambda(M)$ of 3d spacetimes can be expressed in terms of shear-bending coordinates. First, recall that the gravitational symplectic structure  is given by Goldman's symplectic structure \eqref{eq:goldman} with structure group $G_\Lambda$ and the $\Ad$-invariant symmetric bilinear form  $\langle\,,\rangle$ from \eqref{eq:forms}. As the latter is the imaginary part  $\mathrm{Im}_\ell(\kappa)$
of the bilinear extension of the Killing form on $\mathfrak{sl}(2,\R)$   
and the generalised shear coordinates can be interpreted as an $R_\Lambda$-analytic continuation  of the shear coordinates on $\cT(S)$, it is natural to expect that Goldman's symplectic structure on $\cG\cH_\Lambda(M)$ is given 
by a Poisson structure on $R_\Lambda^E$ that resembles the Weil-Petersson Poisson structure.

In the following, we shall prove that the gravitational symplectic structure on $\cG\cH_\Lambda(M)$ is induced from the Poisson bivector 
\begin{align}\label{imagbivec}
\pi_\Lambda=\frac{1}{2}\sum_{\alpha,\beta\in E(\Gamma)}\pi^{\alpha\beta}_{WP}\frac{\partial}{\partial x^\alpha}\wedge \frac{\partial}{\partial y^\beta}
\end{align}
where the coordinates $x^\alpha$ and $y^\beta$ denote, respectively, the real and imaginary parts of the generalised shear coordinates $z^\alpha$ in \eqref{eq:zdef}.
As the Poisson bracket of the variables $x^\alpha, y^\beta$ is a combinatorial constant, it is immediate that the Jacobi identity is satisfied and that this bivector defines a Poisson structure on $R^\Lambda_E$. Moreover, the bivector $\pi_\Lambda$  induces a Poisson structure on the constraint surface $\Ker c_\Lambda  \subset R_\Lambda^E$, where $c_\Lambda: R^E_\Lambda\to R^F_\Lambda$ is the linear constraint map from Theorem \ref{thm:gen_embed}.
This follows directly from the combinatorics of $\pi_\Lambda$, via its relation to $\pi_{WP}$, and from the combinatorics of the constraints.
A simple computation shows that the Poisson bracket of a function
 $f\in C^\infty(R_\Lambda^E)$ with a component $c^i_\Lambda: R^E_\Lambda\to R_\Lambda$ of the constraint map in Theorem \ref{thm:gen_embed} is given by
$$\{f,c^i_\Lambda\}_\Lambda=(df\otimes dc^i_\Lambda)(\pi_\Lambda)=\frac{1}{2}\sum_{\alpha\in E(\Gamma)}\left(\ell\frac{\partial f}{\partial x^\alpha}+\frac{\partial f}{\partial y^\alpha}\right)(\theta^i{}_{\beta}-\theta^i{}_{\gamma}+\theta^i{}_{\delta}-\theta^i{}_{\epsilon}).$$
By an argument similar to the one  following equation \eqref{eq:poiss_const},  one finds that this expression vanishes identically for all $f\in C^\infty(R_\Lambda^E)$. This follows since every face $i\in F(\Gamma)$ involves only left or right turns and hence the linear combinations of the multiplicities $\theta^i$ on the right-hand-side cancel. Consequently, the  constraint $c_\Lambda: R^E_\Lambda\to R^F_\Lambda$ is Casimir with respect to the Poisson structure \eqref{imagbivec} and the Poisson bivector  $\pi_\Lambda$ restricts to a Poisson bivector on $\Ker c_\Lambda\cong \cG\cH_\Lambda(M)$.  It is then easy to see that this Poisson structure is symplectic.

We will now prove that this Poisson structure on $\cG\cH_\Lambda(M)$ agrees with the gravitational symplectic structure, that is, Goldman's symplectic structure \eqref{eq:goldman} for the group $G_\Lambda$ and the $\Ad$-invariant symmetric bilinear form $\langle\,,\,\rangle$ in \eqref{eq:forms}. This yields the following theorem.
\begin{theorem}\label{thm:cotangent}
The linear constraint $c_\Lambda:R^E_\Lambda\to R^F_\Lambda$ defined in \eqref{eq:const_z} is Casimir with respect to the Poisson bivector $\pi_\Lambda$ on $R^E_\Lambda$ and induces a symplectic structure on  $\Ker c_\Lambda\cong \cG\cH_\Lambda(M)$. This symplectic structure coincides with the gravitational Poisson structure on $\cG\cH_\Lambda(M)$.
\end{theorem}
 \begin{proof}
The general idea of the proof is to compare the Poisson structure on $R_\Lambda^E$ induced by \eqref{imagbivec} with Goldman's symplectic structure \eqref{eq:goldman} for a pair of class functions $f,g:G_\Lambda\to\R$. Given $a,b\in\pi_1(S)$ we consider the coordinate expressions of the associated functions $f_a,g_b:\cG\cH_\Lambda(M)\to\R$, defined by $f_a(\rho)=f(\rho(a))$, $g_b(\rho)=g(\rho(b)),$
where $\cG\cH_\Lambda(M)$
is viewed as a component of $\Rep(S,G_\Lambda)$ and $\rho:\pi_1(S)\to G_\Lambda$ is a group homomorphism. 
With the parametrisation \eqref{eq:hol_shear} of this group homomorphism $\rho$ in terms of generalised shear coordinates
we are able compute the Poisson bracket
$$\{f_a,g_b\}_\Lambda=(df_a\otimes dg_b)(\pi_\Lambda).$$
The combinatorial structure of the Poisson bivector $\pi_\Lambda$ then allows one to interpret each non-trivial contribution in terms of the essential intersection points of closed edge paths representing $a$ and $b$.

For this, consider an embedded fat graph $\Gamma$ dual to an ideal triangulation of $S$ and let $(\alpha_1,...,\alpha_n)$, $(\beta_1,...,\beta_m)$ be closed edge paths in $\Gamma$ that are freely homotopic to, respectively,  $a\in \pi_1(S)$ and $b\in\pi_1(S)$. Expression \eqref{eq:hol_shear} allows one to interpret their  holonomies  as analytic functions  $\rho(a),\rho(b):R^E_\Lambda\to G_\Lambda$.  The  associated functions  $f_a,g_b$ and $F_a,G_b$ in \eqref{eq:goldman}  can then be expressed in shear coordinates as
\begin{align*}
&f_a(z)=f(P^a_nE(z^{\alpha_n})\cdots P^a_1E(z^{\alpha_1}))\\
&\Im_\ell\;\kappa(F_a(z),X)=\frac{d}{dt}\Big|_{t=0}f(P^a_nE(z^{\alpha_n})\cdots P^a_1E(z^{\alpha_1})e^{tX}).\nonumber
\end{align*}
The conjugation invariance of $f$ then yields the  identities
\begin{align*}
\frac{\partial f_a}{\partial x^\alpha}=\sum_{k=1}^n\delta^{\alpha_k}_\alpha\Im_\ell\;\kappa(F_a(z),J^a_k(z)),
\qquad \frac{\partial f_a}{\partial y^\alpha}=\sum_{k=1}^n\delta^{\alpha_k}_\alpha \Re_\ell\;\kappa(F_a(z),J^a_k(z)),
\end{align*}
where $J^a_k:R_\Lambda^E\to\mathfrak{g}_\Lambda$ is given by
\begin{align}\label{eq:jabdef}
J^a_k(z)
=& \Ad_{A_k^a(z)}J_1
\end{align}
with $A_k^a$ defined as in \eqref{eq:adef} and $J_1$ as in \eqref{eq:jgen}.
This allows us to directly compute the Poisson bracket of the functions  $f_a,g_b: R^E_\Lambda\to\R$ induced by the bivector $\pi_\Lambda$ in \eqref{imagbivec}
\begin{align}\label{wpbrackcurve}
\{f_a,g_b\}_\Lambda=\Im_\ell\Big(\!\!\!\!\sum_{\alpha_k\in a,\beta_l\in b}\!\!\!\!\pi_{WP}^{\alpha_k\beta_l}\kappa(F_a,J^a_k)\kappa(G_b,J^b_l)\Big).
\end{align}
To show that this agrees with Goldman's symplectic structure \eqref{eq:goldman},  note that \eqref{wpbrackcurve} is obtained as the imaginary part of the contraction of $F_a\oo G_b\in\mathfrak{g}_\Lambda\oo\mathfrak{g}_\Lambda$ with the bivector
\begin{align}\label{bivecid}
\sum_{\alpha_k\in a,\beta_l\in b}\pi^{\alpha_k\beta_l}_{WP}J^a_k \otimes J^b_l\in \mathfrak{g}_\Lambda\wedge\mathfrak{g}_\Lambda
\end{align}
with respect to the $\Ad$-invariant bilinear form $\kappa$ on $\mathfrak{g}_\Lambda$. Similarly, Goldman's symplectic structure \eqref{eq:goldman} is obtained by contracting  $F_a\oo G_b$ with the bivector
\begin{align}\label{gold_bivec}
\sum_{p\in a\cap b}2\epsilon_p(a,b)\Ad_{A^a_p\otimes A^b_p}\Big(\sum_{i,j=0}^2\eta^{ij}J_i \otimes J_j\Big)\in \mathfrak{g}_\Lambda\wedge\mathfrak{g}_\Lambda
\end{align}
where $J_i$  is given  by \eqref{eq:jgen}.
We now compare the bivectors \eqref{gold_bivec} and \eqref{bivecid}, making use of  the combinatorial structure of $\pi_{WP}$.

First, note that $\pi^{\alpha_k\beta_l}_{WP}=0$  unless the edges $\alpha_k$ and $\beta_l$ are distinct edges sharing a common vertex. As all vertices of $\Gamma$ are trivalent, this implies non-zero contributions only arise for either
$\alpha_{k-1}=\beta_{l-1}$, $\alpha_{k-1}=\beta_{l+1}$, $\alpha_{k+1}=\beta_{l-1}$ or $\alpha_{k+1}=\beta_{l+1}$. We may thus organise the sum in \eqref{bivecid} as a sum over edge path segments in the intersection of edge paths $a$ and $b$.
Up to cyclic relabelling of the edges and orientation reversal, each contribution involves segments of the form $(\alpha_2,...,\alpha_{s-1})=(\beta_2,...,\beta_{s-1})$ with distinct initial and final edges, $\alpha_1\neq\beta_1$ and $\alpha_s\neq\beta_s$.
\begin{figure}[!ht]
\centering
\includegraphics[scale=0.52]{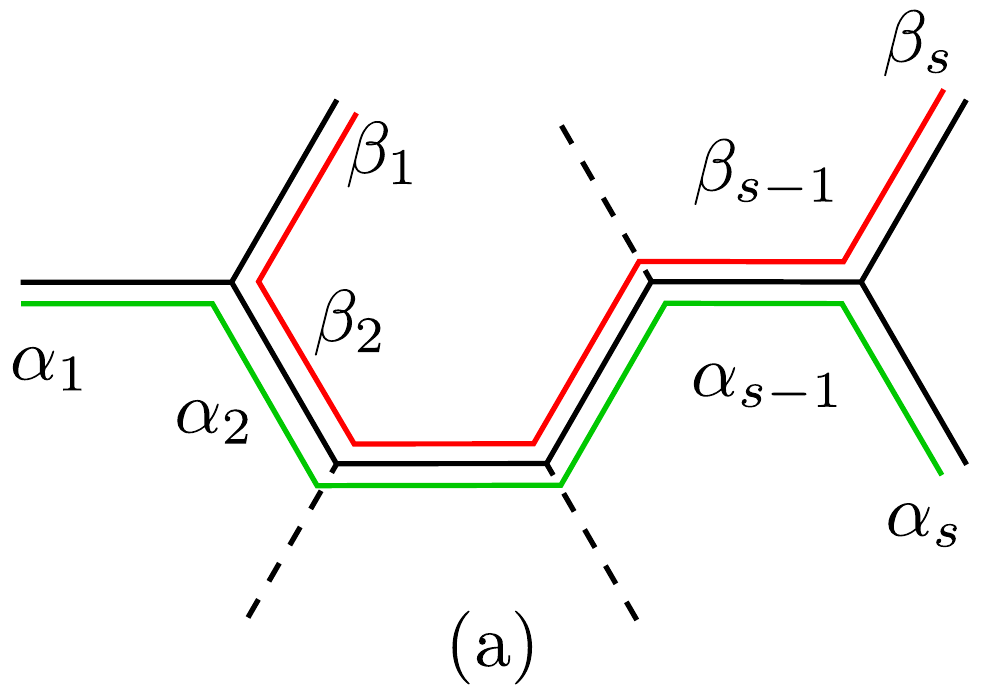}
\qquad\qquad\quad
\includegraphics[scale=0.52]{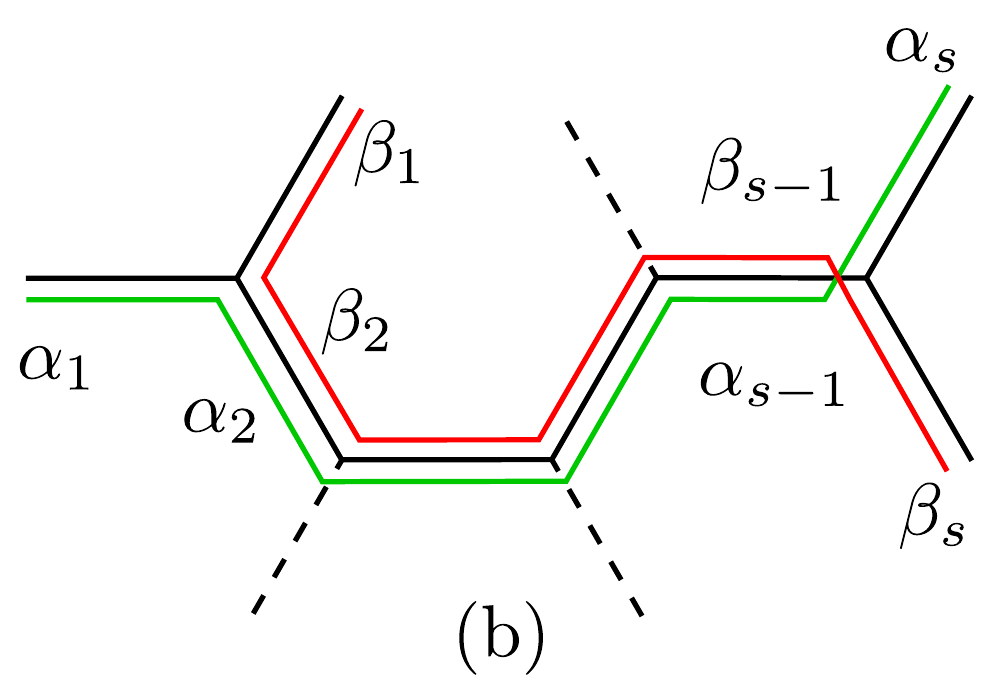}
\caption{Edge path segments contributing to the Poisson bracket.}
\label{fig:edge_sequence}
\end{figure}
We thus obtain
\begin{align*}
\sum_{\alpha_k\in a,\beta_l\in b}\!\!\!\!&\pi^{\alpha_k\beta_l}_{WP}J^a_k \otimes J^b_l=
\!\!\!\!\!\!\sum_{\text{segm.}\in a\cap b}\!\!\!\!\!\!\Big\lbrack\pi_{WP}^{\alpha_1\beta_{2}}\Big(J_1^a\otimes J_{2}^b+J_{2}^a\otimes J_1^b-J^a_1\oo J^b_1 \Big)
\\ \nonumber
&+\sum_{k=2}^{s-2}\,\pi_{WP}^{\alpha_k\beta_{k+1}}\Big(J_k^a\otimes J_{k+1}^b-J_{k+1}^a\otimes J_k^b\Big)
+\pi_{WP}^{\alpha_{s-1}\beta_s}\Big(J_{s-1}^a\otimes J_s^b+J_s^a\otimes J_{s-1}^b-J^a_s\oo J^b_s \Big)\Big\rbrack.
\end{align*} 
To simplify this expression, we use the following relations between the coefficients of the Weil-Petersson bivector
\begin{align*}
\pi_{WP}^{\alpha_k\alpha_{k+1}}=\begin{cases}
1 & \text{if } P_k=L,
\cr
-1 & \text{if } P_k=R,
\end{cases}
\qquad
\pi_{WP}^{\beta_k\beta_{k+1}}=\begin{cases}
-\pi_{WP}^{\alpha_1\alpha_2}=\pi^{\alpha_1\beta_1}_{WP} & \text{for}\;k=1,
\\
\pi_{WP}^{\alpha_k\alpha_{k+1}} & \text{for}\;k=2,...,s-2,
\\
-\pi_{WP}^{\alpha_{s-1}\alpha_s}=-\pi_{WP}^{\alpha_s\beta_s} & \text{for}\;k=s-1,
\end{cases}
\end{align*}
and the following recursion relation satisfied by the Lie algebra-valued functions $J_k^a$
\begin{align*}
J^a_{k}-J^a_{k+1}=\Ad_{A_k^a}(J_0-\pi^{\alpha_k\alpha_{k+1}}_{WP}J_2),
\end{align*}
as well as its counterpart for $J_k^b$. These recursion relations are derived from \eqref{eq:jabdef} by direct computation using the commutators
$$[P_k^a,J_1]
=J_0-\pi^{\alpha_k\alpha_{k+1}}J_2.$$
With these identities, one finds that the Lie algebra-valued bivector \eqref{bivecid} can be expressed as
\begin{align*}
\sum_{\alpha_k\in a,\beta_l\in b}\!\!\!\!\pi^{\alpha_k\beta_l}_{WP}J^a_k \otimes J^b_l&
=\!\!\!\!\!\!\sum_{\text{segm.}\in a\cap b}\!\!\!\!\!\!\Big\lbrack \pi_{WP}^{\alpha_1\beta_{2}}\Ad_{A_2^a\oo A_2^b}\Big(\!\sum_{i,j=0}^2\!\eta^{ij}J_i\oo J_j\Big)
+\pi_{WP}^{\alpha_{s-1}\beta_s}\Ad_{A_{s-1}^a\oo A_{s-1}^b}\Big(\!\sum_{i,j=0}^2\!\eta^{ij}J_i\oo J_j\Big)
\\ &
+\Ad_{A_2^a\oo A_2^b}\Big(J_2\oo J_0-J_0\oo J_2\Big)-\Ad_{A_{s-1}^a\oo A_{s-1}^b}\Big(J_2\oo J_0-J_0\oo J_2\Big)
\\&
+\sum_{k=2}^{s-2}\Ad_{A_k^a\oo A_k^b}\Big(J_1\otimes J_2-J_2\otimes J_1+\pi^{\alpha_k\beta_{k+1}}(J_0\otimes J_1-J_1\otimes J_0)\Big)\Big\rbrack.
\end{align*}
A simple computation based on the relations
$$[P_k,J_0]=J_1P_k+\frac{1}{2}(J_0-\pi^{\alpha_k\alpha_{k+1}}J_2)\qquad [P_k,\pi^{\alpha_k\alpha
_{k+1}}J_2]=J_1P_k+\frac{1}{2}(J_0-\pi^{\alpha_k\alpha_{k+1}}J_2),$$
yields the identities
\begin{align*}
&\Ad_{A_{k+1}^a\oo A_{k+1}^b}\Big(J_2\oo J_0-J_0\oo J_2\Big)-\Ad_{A_k^a\oo A_k^b}\Big(J_2\oo J_0-J_0\oo J_2\Big)
\\
&=\Ad_{A_k^a\oo A_k^b}\Big(J_1\otimes J_2-J_2\otimes J_1+\pi^{\alpha_k\beta_{k+1}}(J_0\otimes J_1-J_1\otimes J_0)\Big)\\
&\Ad_{A_{k+1}^a\oo A_{k+1}^b}\Big(\sum_{i,j=0}^2\!\eta^{ij}J_i\oo J_j\Big)=\Ad_{A_k^a\oo A_k^b}\Big(\sum_{i,j=0}^2\!\eta^{ij}J_i\oo J_j\Big).
\end{align*}
These identities allow one to rewrite the bivector \eqref{bivecid} as
\begin{align*}
\sum_{\alpha_k\in a,\beta_l\in b}\!\!\!\!&\pi^{\alpha_k\beta_l}_{WP}J^a_k \otimes J^b_l
=\!\!\!\!\!\!\sum_{\text{segm.}\in a\cap b}\!\!\!\!\!\!(\pi_{WP}^{\alpha_1\beta_{2}}+\pi_{WP}^{\alpha_{s-1}\beta_s})\Ad_{A_{s-1}^a\oo A_{s-1}^b}\Big(\!\sum_{i,j=0}^2\!\eta^{ij}J_i\oo J_j\Big),
\end{align*}
and the associated
 Poisson bracket \eqref{wpbrackcurve} takes the form
\begin{align*}
\{f_a,g_b\}_\Lambda=\!\!\!\!\sum_{\text{segm.}\in a\cap b}\!\!\!\!(\pi_{WP}^{\alpha_1\beta_{2}}+\pi_{WP}^{\alpha_{s-1}\beta_s})\Im_\ell\Big(\sum_{i,j=0}^2\eta^{ij}\kappa(F_{a_p},J_i)\kappa(G_{b_p},J_j)\Big).
\end{align*}
It  remains to relate the factor $\pi_{WP}^{\alpha_1\beta_{2}}+\pi_{WP}^{\alpha_{s-1}\beta_s}$ for each edge path segment in $a\cap b$ to the oriented intersection numbers of each essential intersection point. For this, note that when $\pi_{WP}^{\alpha_1\beta_{2}}+\pi_{WP}^{\alpha_{s-1}\beta_s}=0$, the edge path segments $(\alpha_1,...,\alpha_s)$ and $(\beta_1,...,\beta_s)$ do not give rise to an essential intersection point, as shown in Figure \ref{fig:edge_sequence} (a). On the other hand, when $\pi_{WP}^{\alpha_1\beta_{2}}=\pi_{WP}^{\alpha_{s-1}\beta_s}$, the edge path segments $(\alpha_1,...,\alpha_s)$ and $(\beta_1,...,\beta_s)$ correspond to exactly one essential intersection point $p$, as shown in Figure \ref{fig:edge_sequence} (b). This proves the identity
$$\pi_{WP}^{\alpha_1\beta_{2}}+\pi_{WP}^{\alpha_{s-1}\beta_s}=2\epsilon_p(a,b),$$
and shows that the bivectors \eqref{bivecid} and \eqref{gold_bivec}  agree.
\end{proof}

Note that the proof of Theorem \ref{thm:cotangent} is largely combinatorial. It only makes use of properties of the Weil-Petersson bivector \eqref{eq:wp_bivect} and of the expression \eqref{eq:hol} for the holonomies in terms of   shear coordinate and the matrices  $E$, $L$ and $R$, defined in Section \ref{sec:background}. It therefore directly generalises the proof that the Weil-Petersson bivector \eqref{eq:wp_bivect} induces the Weil-Petersson symplectic structure on Teichm\"uller space. In fact, this proof is obtained directly from the  proof of Theorem \ref{thm:cotangent} by  replacing $R^E_\Lambda$ with $\R^E$,  $G_\Lambda$ with $\mathrm{PSL}(2,\R)$ and omitting the expressions $\mathrm{Im}_\ell$ and $\mathrm{Re}_\ell$ throughout the proof. 

As a final remark on the symplectic structures on the moduli spaces $\cG\cH_\Lambda(M)$, we  consider  Goldman's symplectic structure for the groups $G=G_\Lambda$ and for a general real linear combination
$B=\mu(\,,\,)+ \nu \langle\,,\,\rangle$
of the $Ad$-invariant symmetric bilinear forms 
 $(\,,\,)$ 
 and $\langle\,,\,\rangle$ 
 in \eqref{formreal}.
A short computation shows that the bilinear form $B$ is non-degenerate if and only if $\Lambda \mu^2+\nu^2\neq 0$. While this condition is satisfied for all non-vanishing linear combinations  if  $\Lambda=1$, it is violated   if $\mu=\pm \nu$ and  $\Lambda=-1$ or   $\nu=\Lambda=0$. 
This shows that Goldman's symplectic structure \eqref{eq:goldman} is well-defined for all values of $\Lambda$ if $\mu=0$ and $\nu=1$, which is the case considered above. In contrast, 
 the choice $\mu=1$ and $\nu=0$ yields $B=(\,,\,)=\mathrm{Re}_\ell(\kappa)$ and is defined only for $\Lambda\neq 0$. In this case, expression \eqref{eq:b_gold} becomes
\begin{align}
&B(F_{a_p}(\rho), G_{b_p}(\rho))=\frac 1 {2\Lambda}\sum_{i,j=0}^{2}\kappa^{ij}\Big((F_{a_p}, J_i)(G_{b_p}, J_j)+ \kappa^{ij}(F_{a_p}, P_i)(G_{b_p}, P_j)\Big).
\end{align}
A direct computation along the lines of the proof of Theorem \ref{thm:cotangent} shows that this corresponds to the choice of the following Poisson bivector on $R^E_\Lambda$
$$
\pi=\frac 1 2 \sum_{\alpha\in E(\Gamma)} \pi^{\alpha\beta}_{WP} \left(\frac{\partial}{\partial x^\alpha}\wedge \frac{\partial}{\partial x^\beta}+\frac 1 \Lambda \frac{\partial}{\partial y^\alpha}\wedge \frac{\partial}{\partial y^\beta}\right).
$$
This provides an additional motivation for the choice of the symmetric bilinear form $\langle\,,\,\rangle=\mathrm{Im}_\ell(\kappa)$ on $\mathfrak g_\Lambda$ when defining the symplectic structure on $\cG\cH_\Lambda(M)$. Besides the physical considerations discussed in Section \ref{subsec:symp}, in particular Theorem \ref{thm:witten}, this choice of bilinear form is the only one which is non-degenerate for all values $\Lambda\in\R$. This allows one to interpret the cosmological constant $\Lambda$ as a deformation parameter in the description of the symplectic structure on the moduli spaces $\cG\cH_\Lambda(M)$.

\subsubsection*{Geometrical interpretation in terms of earthquake and grafting}

To give a geometrical interpretation to the gravitational symplectic structure it is instructive to determine the transformation of the the edge coordinates $z^\alpha$ generated via Poisson brackets by the traces of the $G_\Lambda$-valued holonomies.
Thus, let $(\alpha_1,...,\alpha_n)$ be a closed edge path freely homotopic to 
 $a\in\pi_1(S)$ and denote by $\rho(a)$ the associated holonomy defined as in \eqref{eq:hol_shear} and interpreted as a $G_\Lambda$-valued function over $\cG\cH_\Lambda(M)$. Using the identities
\begin{align}
\ell\frac{\partial\rho(a)}{\partial x^{\alpha_k}}=\frac{\partial\rho(a)}{\partial y^{\alpha_k}}=\ell J_k^a\rho(a),
\end{align}
which follow directly from \eqref{eq:hol_shear} and \eqref{eq:jabdef}, one obtains the expressions for the Poisson bracket between the real and imaginary parts of the holonomies and the shear-bending coordinates
\begin{align*}
&\{\Re_\ell\tr\rho(a),z^\beta\}=\frac{1}{2}\sum_{k=1}^n \pi^{\alpha_k\beta}_{WP}\,\tr(\ell J^a_k\rho(a)),
\qquad
\{\Im_\ell\tr\rho(a),z^\beta\}=\frac{1}{2}\sum_{k=1}^n \pi^{\alpha_k\beta}_{WP}\,\tr(J^a_k\rho(a)).
\end{align*}
This shows that the  flows  generated by, respectively, $\Re_\ell\tr\rho(a)$ and  $\Im_\ell\tr\rho(a)$, take the form 
\begin{align*}
&\Phi_t^{Re}(z^\beta)=z^\beta+\frac{\ell t}{2} \sum_{k=1}^n\pi^{\alpha_k\beta}_{WP}\tr(J^a_k\rho(a))+O(t^2),\quad
&\Phi_{t}^{Im}(z^\beta)=z^\beta+\frac{t}{2} \sum_{k=1}^n\pi^{\alpha_k\beta}_{WP}\tr(J^a_k\rho(a))+O(t^2).
\end{align*}
Computing the transformation of the holonomies \eqref{eq:hol_shear} generated by these Hamiltonian functions then shows that they are related by shear-bending cocycles along $a$ with measures $\ell t$ and $t$, respectively. The transformations generated by the real and imaginary part of $\tr\rho(a)$ can therefore be interpreted in terms of grafting and earthquake along the unique geodesic homotopic to $a$.

The description in terms of shear-bending coordinates thus generalises the result in \cite{Meusburger2007}, where it was shown that grafting and earthquake transformations along  closed simple geodesics are generated via the Atiyah-Bott-Goldman symplectic structure \eqref{eq:goldman} by the real and imaginary part of the traces of the associated holonomy. 
This can also be viewed as a generalisation of the well-known fact that the earthquake map in Teichm\"uller space along a closed simple geodesic $a\in\pi_1(S)$ is generated by the geodesic length $l(a)$,  which is given by the trace of the associated holonomy $\rho(a)\in \mathrm{PSL}(2,\R)$.

\subsection{Cotangent bundle over Teichm\"uller space}\label{subsec:symplec}
\subsubsection*{Shear coordinates description}

To relate the symplectic structure on $\cG\cH_\Lambda(M)$ to the cotangent bundle structure on $T^*\cT(S)$, we describe the latter as a constrained submanifold of $T^*\R^E$ with the  symplectic structure induced by symplectic reduction. For this, note that the global parametrisation of Teichm\"uller space $\cT(S)$ by shear coordinates also provides a global parametrisation of its tangent and cotangent bundles. Given a trivalent fat graph $\Gamma$ dual to an ideal triangulation  $S$,
we may describe the tangent and cotangent bundles over $\cT(S)$ in terms of the coordinate vector fields $\partial/\partial x^\alpha$ and coordinate one-forms $dx^\alpha$ on $\R^E$. 
A tangent vector $\xi=\sum_\alpha \xi^\alpha \;\partial/\partial x^\alpha$ to $\R^E$ at a point in $\Ker\,c$ determines a tangent vector to $\cT(S)$ if and only if its coefficient functions $\xi^\alpha\in C^\infty(\R^E)$ satisfy the constraints \eqref{eq:teich_const}
\begin{align*}
c^i(\xi)=\sum_{\alpha\in E(\Gamma)}\theta^i{}_\alpha\xi^\alpha=0,
\end{align*}
for every face $i\in F(\Gamma)$. By duality, one-forms on $\cT(S)$  correspond to equivalence classes of one-forms on $\R^E$ modulo translations by linear combinations of the differentials of the constraints
\begin{align}\label{eq:p_equiv}
p=\sum_{\alpha\in E(\Gamma)} p_\alpha dx^\alpha\quad\text{with}\quad p\sim p+\sum_{\alpha\in E(\Gamma)}\sum_{i\in F(\Gamma)} p_i\theta^i_\alpha  \,dx^\alpha.
\end{align}
This  implies that the cotangent bundle $T^*\cT(S)\cong T^*\Ker\, c$ is given as the direct product of $\Ker\, c\subset\R^E$, which parametrises the base space, and of the quotient $(\R^E)^*/\mathrm{Ann}(\Ker\, c)$, which parametrises its fibres in terms of equivalence classes of one-forms on $\R^E$. Here and in the following  $\mathrm{Ann}(\Ker\, c)=\mathrm{Span}\{c^1,...,c^F\}$
denotes the annihilator subspace of $\Ker\, c\subset\R^E$.

From the viewpoint of constrained mechanical systems, see for instance \cite{Henneaux1992},  this quotient can be interpreted as a gauge freedom in the definition of coordinates on the fibres of $T^*\cT(S)$, which may be then eliminated via an appropriate gauge fixing condition. This is convenient since it allows one to describe the cotangent bundle $T^*\cT(S)$ as a constrained submanifold of $T^*\R^E$. Thus, note that the constraints $c^i: \R^E\to\R$ in \eqref{eq:teich_const} are first-class with respect to the cotangent bundle symplectic structure defined by the Poisson bivector
\begin{align}\label{eq:cot_bivect}
\pi_{T^*}=\sum_{\alpha\in E(\Gamma)} \frac{\partial}{\partial x^\alpha}\wedge\frac{\partial}{\partial p_\alpha}.
\end{align}
This means the constraints satisfy the relations $\{c^i, c^j\}_{T^*}=dc^i\otimes dc^j(\pi_{T^*})=0$ for all $i,j\in F(\Gamma)$ and that the
equivalence relation in \eqref{eq:p_equiv} is generated by these constraints via the cotangent bundle Poisson bracket $\sum_ip_i\theta^i_{\alpha}=\{\sum_ip_ic^i, p_\alpha\}_{T^*}$.
This allows one to interpret the equivalence classes of one-forms in \eqref{eq:p_equiv} as gauge orbits generated by the constraint $c:\R^E\to\R^F$.

To fix the arbitrary parameters $p_i$ we may then impose  gauge fixing conditions, which can be chosen as  linear constraints on the  coordinates $p_\alpha$
\begin{align*}
\tilde c:(\R^E)^* \to(\R^F)^*,\qquad \tilde c_i(p)=\sum_{\alpha\in E(\Gamma)}\tilde\theta_i{}^\alpha p_\alpha=0.
\end{align*}
The condition that $ \Ker\, \tilde c$ contains 
 exactly one representative in each equivalence class is  equivalent to the invertibility of  the matrix $M\in\mathrm{Mat}(F,\R)$ with entries
\begin{align}\label{eq:dmatrix}
M_i{}^{j}=\{\tilde c_i, c^j\}_{T^*}=\sum_{\alpha\in E(\Gamma)} \tilde \theta_i^\alpha\theta^j_\alpha.
\end{align}
Note that there is a particularly natural choice for $\tilde c$ given by $\tilde\theta=\theta^T$, the transpose of the matrix $\theta$ in \eqref{eq:teich_const}. In the following, however,  we consider more general gauge fixings  which satisfy this condition and refer to such gauge fixing conditions as admissible gauge fixings.  We thus have the following statement.
\begin{proposition}\label{thm:tan_shear}
The cotangent bundle $T^*\cT(S)$ is isomorphic to the quotient 
\begin{align}\label{tquot}
T^*\cT(S)\cong \Ker c\times\frac{(\R^E)^*}{\mathrm{Ann}(\Ker c)},
\end{align}
and
for any admissible gauge fixing map $\tilde c: (\R^E)^*\to(\R^F)^*$, the coordinate functions $x^\alpha,p_\alpha:T^*\cT(S)\to\R$ define an embedding $(x,p):T^*\cT(S)\hookrightarrow\R^E\times(\R^E)^*$ whose image agrees with the kernel of the linear map 
$c\oplus\tilde c:\R^E\times(\R^E)^*\to\R^F\times(\R^F)^*$.
\end{proposition}
We now show that the cotangent bundle symplectic structure on $T^*\R^E\cong\R^E\times(\R^E)^*$  induces the cotangent bundle symplectic structure on $T^*\cT(S)$. In terms of the quotient \eqref{tquot}, this follows from symplectic reduction of $\R^E\times(\R^E)^*$ with respect to the  constraint $c:\R^E\to\R^F$ and means that \eqref{tquot} inherits a symplectic structure that coincides with the cotangent bundle symplectic structure on $T^*\cT(S)$. Equivalently, in terms of gauge fixing, the proof amounts to the construction of the Dirac bracket on $\R^E\times (\R^E)^*$. This is a Poisson structure on $\R^E\times (\R^E)^*$ for which all components of the constraint $c:\R^E\to\R^F$ and the gauge fixing conditions $\tilde c: (\R^E)^*\to(\R^F)^*$ are Casimir functions and which coincides with the original Poisson structure \eqref{eq:cot_bivect} for all functions 
which Poisson-commute with the constraints and gauge fixing conditions.
\begin{proposition}
\label{thm:cot_reduc} $\quad$
\begin{enumerate}
\item The cotangent bundle symplectic structure on $\R^E\times (\R^E)^*$ induces a symplectic structure on the quotient \eqref{tquot} which coincides with the cotangent symplectic structure on $T^*\cT(S)\cong T^*\Ker\, c$.
\item For any admissible gauge fixing $\tilde c: (\R^E)^*\to(\R^F)^*$ 
the linear constraint $c\oplus \tilde c$ 
is second-class with respect to the cotangent bundle symplectic structure on $\R^E\times (\R^E)^*$, and the associated Dirac bracket induces the  cotangent bundle symplectic structure on $T^*\cT(S)\cong \Ker (c\oplus\tilde c)$. 
\end{enumerate}
\end{proposition}
\begin{proof}

The first point is a direct consequence of the theory of linear symplectic reduction, for an accessible overview see \cite{Marsden1994}.
The linear subspace $\Ker c\times (\R^E)^*\subset \R^E\times (\R^E)^*$ is coisotropic  with respect to the cotangent bundle symplectic structure \eqref{eq:cot_bivect}  with symplectic complement $\Ann (\Ker c)=\text{Span}\{c^1,...,c^F\}\subset (\R^E)^*$. Hence, the associated symplectic quotient is given by \eqref{tquot}, and it is immediate from the discussion above that the induced symplectic structure is the cotangent bundle symplectic structure on $T^*\cT(S)$.

For the second point, a direct computation shows that the Dirac matrix for the constraint function $c\oplus\tilde c$ in  Theorem \ref{thm:tan_shear} takes the form
\begin{align}
D=\left(\{(c\oplus\tilde c)_i, (c\oplus \tilde c)_j\}_{T^*}\right)=\left(\begin{array}{cc} 0 & -M^T\\ M & 0
\end{array}\right)
\end{align}
with $M$ given by \eqref{eq:dmatrix}. If the gauge fixing condition $\tilde c$ is admissible, the matrix $M$ is invertible, which implies that $D$ is invertible and the constraint function $c\oplus\tilde c$ is second-class.  This defines a Poisson structure on $\R^E\times (\R^E)^*$, the Dirac bracket, given by
 \begin{align*}
 \{f,g\}_D=\{f,g\}_{T^*}+\sum_{i,j=1}^{2F} D^\inv_{ij}\{f, (c\oplus\tilde c)_i\}_{T^*}\{g, (c\oplus\tilde c)_j\}_{T^*}\qquad\forall f,g\in C^\infty(\R^E\times (\R^E)^*),
\end{align*}
for which all constraint components $(c\oplus\tilde c)_i: \R^E\times (\R^E)^*\to\R$ are Casimir functions. Hence, this Poisson bracket induces a symplectic structure on $\ker (c\oplus \tilde c)\cong T^*\cT(S)$ which coincides with the symplectic structure \eqref{eq:cot_bivect} if  $f$ or $g$ Poisson commute with all  constraint functions $(c\oplus\tilde c)_i$. Moreover, it is  easy to see from the block diagonal form of the Dirac matrix that this symplectic structure coincides with the cotangent bundle symplectic structure on $T^*\cT(S)$.
\end{proof}

\subsubsection*{Symplectomorphisms between $T^*\cT(S)$ and $\cG\cH_\Lambda(M)$}

For all values of the cosmological constant $\Lambda$, the realisation of the cotangent bundle $T^*\cT(S)$ given in Proposition \ref{thm:cot_reduc} allows one to relate the moduli spaces $\cG\cH_\Lambda(M)$ of three-dimensional  MGH Einstein spacetimes with their gravitational symplectic structures to the cotangent bundle $T^*\cT(S)$ with the cotangent bundle structure. From a physics perspective this is motivated by another formulation of 3d gravity as a Hamiltonian system on Teichm\"uller space \cite{Moncrief1989,Krasnov2007} and is quite natural mathematically in view of the common parametrisation of the gravitational moduli spaces by measured laminations \cite{Mess2007,Scannell1999} related to the so-called canonical Wick rotation-rescaling theory \cite{Benedetti2009}. See also \cite{Scarinci2013} for a description of the symplectomorphism $T^*\cT(S)\to\cG\cH_{-1}(M)$ in the more general context of universal Teichm\"uller theory and \cite{Scarinci} for a geometric description of the symplectic properties of 
Wick rotations between the moduli spaces of 3-dimensional geometric structures in relation to earthquakes and harmonic maps between surfaces.

To construct a symplectomorphism  between $T^*\cT(S)$ and $\cG\cH_\Lambda(M)$ in term of shear-bending coordinates, note that the expression \eqref{imagbivec} for the gravitational Poisson bivector $\pi_\Lambda$ on $R_\Lambda^E$ can be readily related to the Poisson bivector $\pi_{T^*}$ given in \eqref{eq:cot_bivect}. Thus we consider the following map between $T^*\R^E$ and $R_\Lambda^E$
\begin{align}\label{pipullb}
\pi_{WP}^\sharp:\R^E\times(\R^E)^*\to\R^E\times\R^E,\qquad (x^\alpha,p_\alpha)\mapsto (x^\alpha,y^\alpha)=(x^\alpha,\Sigma_\beta \pi_{WP}^{\alpha\beta}p_\beta).
\end{align}
It is easy to see that this map is a Poisson map up to a multiplicative constant
$$(\pi_{WP}^\sharp)_*\pi_{T^*}=\sum_{\alpha\in E(\Gamma)} (\pi_{WP}^\sharp)_*\frac{\partial}{\partial x^\alpha}\wedge(\pi_{WP}^\sharp)_*\frac{\partial}{\partial p_\alpha}=-\sum_{\alpha\in E(\Gamma)} \pi^{\alpha\beta}_{WP}\frac{\partial}{\partial x^\alpha}\wedge\frac{\partial}{\partial y_\beta}=-2\pi_{\Lambda},$$
and, together with Proposition \ref{thm:cot_reduc}, that it descends to a symplectomorphism between $T^*\cT(S)$ and $\cG\cH_\Lambda(M)$. 
\begin{theorem} \label{thm:quotient} The linear Poisson map $\pi_{WP}^\sharp:T^*\R^E\to R^E_\Lambda$ defined in \eqref{pipullb} induces a symplectomorphism  $\pi_{WP}^\sharp: T^*\cT(S)\to \cG\cH_\Lambda(M)$.
\end{theorem}
\begin{proof}
By Propositions \ref{thm:tan_shear} and \ref{thm:cot_reduc}, the cotangent bundle $T^*\cT(S)$ can be identified with the  symplectic quotient \eqref{tquot}. By Theorem \ref{thm:cotangent}, the moduli space $\cG\cH_\Lambda(M)$ is given by the restriction of the Poisson structure \eqref{imagbivec} to the  kernel constraint $c_\Lambda: R^E_\Lambda\to R^F_\Lambda$ in \eqref{eq:const_z}.
It is therefore sufficient to show that the linear map $\pi_{WP}^\sharp$ annihilates the linear subspace $\text{Ann}(\Ker c)=\text{Span}\{c^1,...,c^F\}\subset (\R^E)^*$ and maps the linear subspace $\ker c\times (\R^E)^*\subset \R^E\times (\R^E)^*$ to $\Ker c_\Lambda\subset \R^E_\Lambda$.

Both statements follow directly from the fact that the 
constraints on the shear coordinates  for each face  $i\in F(\Gamma)$ are Casimir functions for the Weil-Petersson Poisson bivector, see  Theorem \ref{thm:teich_reduc} and the preceding discussion. This proves that  the map $\pi^\sharp_{WP}$ descends to a symplectomorphism $\pi^\sharp_{WP}: T^*\cT(S)\to\cG\cH_\Lambda(M)$.
\end{proof}

\section{Mapping class group actions}

In this section we investigate the mapping class group action on the moduli spaces $\cG\cH_\Lambda(M)$. We start by describing the action of $\Mod(S)$ on the space of measured geodesic laminations and its $R_\Lambda$-extensions making use of formula \eqref{eq:udef} and the Whitehead moves \eqref{eq:flip_x} for shear coordinates on $\cT(S)$. We then show that the Whitehead moves take a particularly simple form in shear-bending coordinates on $\cG\cH_\Lambda(M)$, which can be viewed as an analytic continuation of formula \eqref{eq:flip_x}. We then prove that these Whitehead moves induce three different actions of the mapping class group on the cotangent bundle $T^*\cT(S)$ of Teichm\"uller space, corresponding to the different values of $\Lambda$. Finally, we prove that all these mapping class group actions are symplectic, making use of simple decomposition of the Whitehead move transformation into a non-linear term generated by the Poisson structure and a linear term which implements the combinatorial 
transformation of the Poisson structure under the Whitehead move.

\subsection{Mapping class group action on $\cG\cH_\Lambda(M)$}\label{subsec:map_shear}
\subsubsection*{Whitehead moves for measured laminations}
We now describe the action of the mapping class group $\Mod(S)$ on $\cM\cL(S)$ in terms of the parametrisation \eqref{eq:udef} by differences of shear coordinates on $\cT(S)$. More generally, we look at the coordinates $w^\alpha$ on $\cM\cL^{\Lambda}(S)$ obtained from \eqref{eq:udef} via analytic continuation. Similarly to the derivation in Subsection \ref{subsec:back_mcg} of the shear coordinate expression \eqref{eq:flip_x} for the mapping class group action on Teichm\"uller space, we extend the action on $\cM\cL^\Lambda(S)$ to $\R^E\times \R^E$ by equivariance
$$(x'{}^\alpha,w'{}^\alpha)=\varphi(x^\alpha,w^\alpha)=(x^\alpha,w^\alpha)\circ\varphi^*.$$
Here $x^\alpha$ and $w^\alpha$ are the shear coordinates on the base Teichm\"uller space and on the fibres of $R_\Lambda$-valued measured laminations associated with an embedded trivalent fat graph $\Gamma$ on $S$ and $x'{}^\alpha$ and $w'{}^\alpha$ the corresponding  coordinates associated with $\Gamma'=\varphi(\Gamma)$ for $\varphi\in\Mod(S)$.

For a Whitehead move along an edge $\alpha\in E(\Gamma)$, it is easy to compute the coordinate transformation, as the coordinates $w^\alpha$ are analytic extensions of differences of coordinates $x^\alpha$. This yields 
\begin{align}\label{eq:flip_lam}
W_\alpha:\begin{cases}
x^\alpha \mapsto x'{}^{\alpha}=-x^\alpha,
\cr
x^{\beta,\delta}\mapsto x'{}^{\beta,\delta}=x^{\beta,\delta}+\log(1+e^{x^\alpha}),
\cr
x^{\gamma,\epsilon}\mapsto x'{}^{\gamma,\epsilon}=x^{\gamma,\epsilon}-\log(1+e^{-x^\alpha}),
\cr
w^\alpha\mapsto w{}'^\alpha=-w^\alpha,
\cr
w^{\beta,\delta}\mapsto w'{}^{\beta,\delta}=w^{\beta,\delta}+\log\Big(\frac{1+e^{x^\alpha+w^\alpha}}{1+e^{x^\alpha}}\Big),
\cr
w^{\gamma,\epsilon}\mapsto w'{}^{\gamma,\epsilon}=w^{\gamma,\epsilon}-\log\Big(\frac{1+e^{-x^{\alpha}-w^{\alpha}}}{1+e^{-x^{\alpha}}}\Big),
\end{cases}
\end{align}
while the coordinates of all other edges are preserved. A direct computation, which again makes use of the definition of the coordinates $w^\alpha$ as an analytic continuation of differences of shear coordinates, 
shows that the Whitehead moves \eqref{eq:flip_lam} satisfy all relations of Theorem  \ref{thm:mcg_fat_graph} and also preserve the analytic continuation of the constraints \eqref{eq:u_const}. This proves the following proposition.

\begin{proposition}\label{lem:map_lam}
The transformations \eqref{eq:flip_lam} induce an action of the mapping class group $\Mod(S)$ on the bundle $\cM\cL(S)$ of measured geodesic laminations on $S$.
\end{proposition}

\subsubsection*{The mapping class group action in shear-bending coordinate}

We now consider the moduli spaces of MGH Einstein spacetimes. The definition \eqref{eq:zdef} of the  shear-bending coordinates on $\cG\cH_\Lambda(M)$ now allows one to derive their transformation under Whitehead moves directly from the action \eqref{eq:flip_lam}  by adding the  terms in \eqref{eq:flip_lam} for the real and imaginary part of the coordinates $z^\alpha$
\begin{align}\label{map_uni}
W_\alpha^\Lambda:\begin{cases}z^\alpha\mapsto z'^\alpha=-z^\alpha,\\
z^{\beta,\delta}\mapsto z'^{\beta,\delta}=z^{\beta,\delta}+\log(1+e^{z^\alpha}),\\
z^{\gamma,\epsilon}\mapsto z'^{\gamma,\epsilon}=z^{\gamma,\epsilon}-\log(1+e^{-z^\alpha}).
\end{cases}
\end{align}
Note that for  $\Lambda=1$ the logarithms are not well defined due to the presence of branching points. These, however, does not affect the holonomies where only the exponentials of shear coordinates appear.

Clearly, the transformation \eqref{map_uni} is simply the analytic continuation of the associated Whitehead move $W_\alpha:\R^E\to\R^E$ in \eqref{eq:flip_x}. As an immediate generalisation of Theorem \ref{thm:pentagon_teich} we thus have the following result.
\begin{theorem}\label{thm:pentagon_shear_bend}
The Whitehead moves $W_\alpha^\Lambda:R_\Lambda^E\to R_\Lambda^E$ satisfy the relations of Theorem \ref{thm:mcg_fat_graph}. Furthermore, they preserve the constraints \eqref{eq:const_z} and induce an action of the mapping class group $\Mod (S)$ on $\cG\cH_\Lambda(M)$.
\end{theorem}
\begin{proof}
This follows immediately from Theorem \ref{thm:pentagon_teich} as the Whitehead moves \eqref{map_uni} and the constraints \eqref{eq:const_z} are the analytic continuation of \eqref{eq:flip_x} and \eqref{eq:teich_const}. The computation proving the pentagon identity for Whitehead moves \eqref{eq:flip_x} is given in Appendix \ref{app:flip} for the convenience of the reader. The other relations of Theorem \ref{thm:mcg_fat_graph} are easily verified.
\end{proof}

\subsubsection*{Mapping class group actions on $T^*\cT(S)$}

In view of the symplectomorphism between $T^*\cT(S)$ and $\cG\cH_\Lambda(M)$ obtained in Theorem \ref{thm:quotient}, it is natural to study the mapping class group action on $T^*\cT(S)$ induced by \eqref{map_uni} via pull-back. We will now show that the induced actions on $T^*\cT (S)$ are all symplectic but are all distinct for different signs  of the curvature $\Lambda$. This should have interesting consequences for the quantum theory as it provides a common representation for the corresponding quantum operators in terms of  the Weyl algebra such that  the algebra of quantum symmetries, e.~g.~the quantum mapping class group action, is the only distinguishing feature between the theories for different values of $\Lambda$. This is further evidence for the importance of the mapping class group in the quantisation.

We start by computing  the pull-back of the Whitehead move $W^\Lambda_\alpha:R_\Lambda^E\to R_\Lambda^E$ \eqref{map_uni} with the linear Poisson map $\pi^\sharp_{WP}: T^*\R^E\to R_\Lambda^E$ in \eqref{pipullb}
$$W^\Lambda_\alpha\circ \pi^\sharp_{WP}:\begin{cases}
(x^\alpha,p_\alpha)\mapsto z'^\alpha=-x^\alpha-\ell\pi^{\alpha\zeta}_{WP} p_\zeta,
\cr
(x^{\beta,\delta},p_{\beta,\delta})\mapsto z'^{\beta,\delta}=x^{\beta,\delta}+\ell\pi^{\beta,\delta\zeta}_{WP}p_\zeta+\log\Big(1+e^{x^\alpha+\ell \pi_{WP}^{\alpha\zeta} p_\zeta}\Big),
\cr
(x^{\gamma,\epsilon},p_{\gamma,\epsilon})\mapsto z'^{\gamma,\epsilon}=x^{\gamma,\epsilon}+\ell\pi^{\gamma,\epsilon\zeta}_{WP} p_\zeta-\log\Big(1+e^{-x^\alpha-\ell\pi^{\alpha\zeta}_{WP} p_\zeta}\Big).
\end{cases}$$
Here and in the following we use Einstein's summation convention and omit the sum over the edge label $\zeta$  for  better legibility. 
Comparing this expression with the map $\pi'{}^\sharp_{\!\!WP}:T^*\R^E\to R_\Lambda^E$ defined by the Weil-Petersson bivector $\pi'_{WP}$ for the fat graph $\Gamma'=\Gamma_\alpha$
it is easy to identify the associated maps $\tilde W^\Lambda_\alpha:T^*\R^E\to T^*\R^E$ for which the following diagram commutes
\begin{equation*}
\xymatrix{ T^*\R^E \ar[dd]_{\widetilde W^\Lambda_\alpha} \ar[rr]^{\;\;\pi^\sharp_{WP}\quad} &  &  R^E_\Lambda \ar[dd]^{W^\Lambda_\alpha} \\ & &  \\
T^*\R^E \ar[rr]_{\;\;\pi'{}^\sharp_{\!\!WP}\quad} & &  R^E_\Lambda.}
\end{equation*}
Note that since $\pi'{}^\sharp_{\!\!WP}$ is not injective, such  maps are unique only up to translations by elements of $\Ker\,\pi'{}^\sharp_{\!\!WP}$.
A direct computation shows that the following map satisfies this condition
\begin{equation}\label{gen_shear_pb}
\tilde W_\alpha^\Lambda:
\begin{cases}
x^\alpha\mapsto x'^\alpha=-x^\alpha,
\cr
x^{\beta,\delta}\mapsto x'^{\beta,\delta}=x^{\beta,\delta}+\Re_\ell \log\Big(1+e^{x^\alpha+\ell \pi_{WP}^{\alpha\zeta} p_\zeta}\Big),
\cr
x^{\gamma,\epsilon}\mapsto x'^{\gamma,\epsilon}=x^{\gamma,\epsilon}-\Re_\ell\log\Big(1+e^{-x^\alpha-\ell\pi^{\alpha\zeta}_{WP} p_\zeta}\Big),
\cr
p_\alpha\mapsto p'_\alpha=-p_\alpha+p_\gamma+p_\epsilon+\Im_\ell \log\left(1+e^{x^\alpha+\ell\pi^{\alpha\zeta}_{WP} p_\zeta}\right),
\cr
p_{\beta,\gamma,\delta,\epsilon}\mapsto p'_{\beta,\gamma,\delta,\epsilon}=p_{\beta,\gamma,\delta,\epsilon}.
\end{cases}
\end{equation}
Using this expression for the map $\tilde W^\Lambda_\alpha:T^*\R^E\to T^*\R^E$, one can verify its properties  by direct computations, which yields the following theorem.

\begin{theorem}\label{thm:mod_cot}
The Whitehead moves $\widetilde W^\Lambda_\alpha:T^*\R^E\to T^*\R^E$ satisfy the relations of Theorem \ref{thm:mcg_fat_graph}. Furthermore, they preserve the constraints \eqref{eq:teich_const} and their gauge orbits and induce an action of the mapping class group $\Mod(S)$ on $T^*\cT(S)$.
\end{theorem}
\begin{proof}
That map $\tilde W^\Lambda_\alpha:T^*\R^E\to T^*\R^E$ in \eqref{gen_shear_pb} satisfy the relations of Theorem \ref{thm:mcg_fat_graph} follows by direct computations. The only non-trivial case is that of the pentagon relation which is analogous to that for \eqref{eq:flip_x}, see Appendix \ref{app:flip}.

That the constraints \eqref{eq:teich_const} and their gauge orbits are preserved can be seen from the combinatorial relation between $\Gamma$ and $\Gamma'=\Gamma_\alpha$. Consider a face $i\in F(\Gamma)$ containing the sequence of edges $(\beta,\alpha,\epsilon)$ and its transformation under a Whitehead move as  in Figures \ref{fig:adj_edges} and \ref{fig:whitehead}. The corresponding face $i\in F(\Gamma')$ then necessarily contains the sequence $(\beta,\epsilon)$. As the constraints are given by a sum over the  coordinates of edges in  the face, it is clear from \eqref{gen_shear_pb} that the pull-back of $c'{}^i$ can be written as
\begin{align*}
c'{}^i\circ \tilde W_\alpha^\Lambda&=x^\beta+\Re_\ell\log\Big(1+e^{x^\alpha+\ell\pi^{\alpha\zeta}_{WP} p_\zeta}\Big)+x^\epsilon-\Re_\ell\log\Big(1+e^{-x^\alpha-\ell\pi^{\alpha\zeta}_{WP} p_\zeta}\Big)+\cdots
\cr
&=x^\beta+x^\epsilon+\Re_\ell\log\Big(\frac{1+e^{x^\alpha+\ell\pi^{\alpha\zeta}_{WP} p_\zeta}}{1+e^{-x^\alpha-\ell\pi^{\alpha\zeta}_{WP} p_\zeta}}\Big)+\cdots=x^\beta+x^\epsilon+x^\alpha+\cdots=c^i,
\end{align*}
where the dots stand for the contribution of the other edges in the face $i$, which is invariant under $\tilde W_\alpha^\Lambda$.
On the other hand, the gauge orbits are the orbits of points in $\R^E\times (\R^E)^*$  under the translations of  one-forms by differentials of the constraints as in 
\eqref{eq:p_equiv}.
So consider the transformation of a one-form $p+qdc^i$ under \eqref{gen_shear_pb}. For the coordinate of the edge $\alpha$ we have
\begin{align*}
\tilde W^\Lambda_\alpha(p_\alpha+q\theta^i{}_\alpha)&=-p_\alpha-q\theta^i{}_\alpha+p_\gamma+q\theta^i{}_\gamma+p_\epsilon+q\theta^i{}_\epsilon+\Im_\ell\log\Big(1+e^{x^\alpha+\ell\pi^{\alpha\zeta}p_\zeta+\ell\pi^{\alpha\zeta}q\theta^i{}_\zeta}\Big)
\cr
&=\tilde W_\alpha^\Lambda(p_\alpha)+q(-\theta^i{}_\alpha+\theta^i{}_\gamma+\theta^i{}_\epsilon)=\tilde W_\alpha^\Lambda(p_\alpha)+q\theta'{}^i{}_\alpha,
\end{align*}
since $\theta^i{}_\alpha=1,\theta^i{}_\gamma=0,\theta^i{}_\epsilon=1$ and $\theta'{}^i{}_\alpha=0$. For any other
 edge, including the neighbouring edges of $\alpha$, we have
\begin{align*}
\tilde W_\alpha^\Lambda(p_{\eta}+q\theta^i{}_{\eta})=p_{\eta}+q\theta^i{}_{\eta}=\tilde W_\alpha^\Lambda(p_{\eta})+q\theta'{}^i{}_{\eta},
\end{align*}
since the multiplicity of $\eta$ on the face $i$ does not change under a Whitehead move. This shows the gauge orbits of $c^i$ are mapped to gauge orbits of $c'{}^i$ or, equivalently,
$$\tilde W_\alpha^\Lambda(p+qdc^i)=\tilde W_\alpha^\Lambda(p)+qdc'{}^i.$$
The same arguments can be applied to other (combinations of)  edge sequences involved in the Whitehead move. This shows that  $c'\circ \tilde W_\alpha^\Lambda=c$ and $\tilde W_\alpha^\Lambda([p])=[\tilde W_\alpha^\Lambda(p)]$ and completes the proof.
\end{proof}

\subsection{The mapping class group action is symplectic}
\label{subsec:ham}
\subsubsection*{Hamiltonians for the Whitehead move}
In this section, we show that the transformations \eqref{map_uni} and \eqref{gen_shear_pb} are Poisson maps with respect to the gravitational and cotangent bundle Poisson structure and that the induced mapping class group actions $\cG\cH_\Lambda(M)$ and $T^*\cT(S)$ are all symplectic. This is achieved by decomposing the Whitehead moves into terms generated via the Poisson structure by a  Hamiltonian and an additional  term, which is linear, purely combinatorial and  independent of $\Lambda$. 
We show that the latter corresponds to the transformation of the Poisson bivector.
Such a decomposition of the Whitehead moves is already possible in the Teichm\"uller context, so the proof presented here can be viewed as a generalisation of the corresponding result for shear coordinates.

We therefore start by considering the Whitehead move for shear coordinates on Teichm\"uller space. Let $\Gamma$ be an embedded trivalent fat graph on $S$ and denote by $x^\alpha$  the corresponding shear coordinate for an edge $\alpha\in E(\Gamma)$. For each edge $\alpha\in E(\Gamma)$, consider the Hamiltonian function $H_\alpha:\R^E\to \R$ given by
\begin{align}\label{eq:falpha}
H_\alpha(x)=H(x^\alpha)=\frac{(x^\alpha)^2}{4}+\Li_2(-e^{x^\alpha}),
\end{align}
where $\Li_2$ denotes Euler's dilogarithm. A short computation using the combinatorial structure of $\pi_{WP}$  and the expression for the derivative of $H_\alpha$
$$
\frac{\partial H_\alpha}{\partial x^\alpha}=\tfrac{1}{2}x^\alpha-\log(1+e^{x^\alpha}).
$$
then shows that 
  the Whitehead move  $W_\alpha:\R^E\to\R^E$ in \eqref{eq:flip_x} is given by
\begin{align*}
W_\alpha:
\begin{cases}
x^\alpha \mapsto 
x'{}^{\alpha}=-x^\alpha
\cr
x^{\beta,\gamma,\delta,\epsilon}\mapsto
x'{}^{\beta,\gamma,\delta,\epsilon}=x^{\beta,\gamma,\delta,\epsilon}+\frac 1 2 {x^\alpha}+\Big\{x^{\beta,\gamma,\delta,\epsilon},H_\alpha\Big\}_{WP}.
\end{cases}
\end{align*}
This allows us to decompose it as  $W_\alpha=A_\alpha\circ B_\alpha$ with maps  $A_\alpha,B_\alpha:\R^E\to\R^E$ defined by
\begin{align*}
A_\alpha:
\begin{cases}
x^\alpha \mapsto 
x'{}^{\alpha}=-x^\alpha,
\cr
x^{\beta,\gamma,\delta,\epsilon}\mapsto
x'{}^{\beta,\gamma,\delta,\epsilon}=x^{\beta,\gamma,\delta,\epsilon}+\frac 1 2 {x^\alpha},
\end{cases}
B_\alpha:
\begin{cases}
x^\alpha \mapsto 
x'{}^{\alpha}=x^\alpha,
\cr
x^{\beta,\gamma,\delta,\epsilon}\mapsto
x'{}^{\beta,\gamma,\delta,\epsilon}=x^{\beta,\gamma,\delta,\epsilon}+\{x^{\beta,\gamma,\delta,\epsilon},H_\alpha\}_{WP}.
\end{cases}
\end{align*}
It is immediate that the transformation $B_\alpha$ is a Poisson map  $B_\alpha:(\R^E,\pi_{WP})\to(\R^E,\pi_{WP})$, as it can be interpreted as the Hamiltonian flow generated  by the Hamiltonian $H_\alpha$ via the Weil-Petersson Poisson bracket. In contrast, the combinatorial linear transformation $A_\alpha: \R^E\to\R^E$ 
transforms the Poisson bivector $\pi_{WP}$ associated with the fat graph  $\Gamma$ into the Poisson bivector $\pi_{WP}'$ associated with $\Gamma'=\Gamma_\alpha$. This follows by a simple computation of the push-forward of $\pi_{WP}$ by $A_\alpha$, which yields
$(A_\alpha)_*\pi_{WP}=\pi'_{WP}$.
This gives a simple proof that the Whitehead moves \eqref{eq:flip_x} indeed induces a symplectic action of $\Mod(S)$ on $\cT(S)$, as stated in Theorem \ref{thm:pentagon_teich}.

Similarly, in the context of the moduli spaces $\cG\cH_\Lambda(M)$, we may decompose the Whitehead moves $W^\Lambda_\alpha:R_\Lambda^E\to R_\Lambda^E$ in \eqref{map_uni} using the analytic extension $H^\Lambda_\alpha:R_\Lambda\to R_\Lambda$ of the Hamiltonian function \eqref{eq:falpha}. More precisely, for each edge $\alpha\in E(\Gamma)$ consider the real functions $\Im_\ell H^\Lambda_\alpha:R_\Lambda\to\R$ given by the imaginary part of $H_\alpha^\Lambda$
\begin{align*}
\mathrm{Im}_\ell\,H^\Lambda_\alpha(z)=\mathrm{Im}_\ell\left(H^\Lambda(z^\alpha)\right)=\tfrac 1 2 x^\alpha y^\alpha +\begin{cases}
-y^\alpha\log(1+e^{x^\alpha}),  & \Lambda=0,\\
 \tfrac 1 2 \Li_2(-e^{x^\alpha+y^\alpha})-\tfrac 1 2 \Li_2(-e^{x^\alpha-y^\alpha}), &  \Lambda=-1, \\
 \mathrm{Im}(\Li_2(-e^{x^\alpha+\mathrm{i} y^\alpha})), & \Lambda=1,
\end{cases}
\end{align*}
and their derivatives with respect to the real and imaginary parts of $z^\alpha$
\begin{align*}
\frac{\partial \mathrm{Im}_\ell H^\Lambda_\alpha}{\partial x^\alpha}=\mathrm{Im}_\ell\Big(\tfrac{1}{2}z^\alpha-\log(1+e^{z^\alpha})\Big),
\qquad \frac{\partial\mathrm{Im}_\ell H^\Lambda_\alpha}{\partial y^\alpha}=\mathrm{Re}_\ell\Big(\tfrac{1}{2}z^\alpha-\log(1+e^{z^\alpha})\Big).
\end{align*}
The Whitehead move $W^\Lambda_\alpha:R_\Lambda^E\to R_\Lambda^E$ can then  be written as
\begin{align*}
W_\alpha^\Lambda:\begin{cases}z^\alpha\mapsto z'^\alpha=-z^\alpha,
\\
z^{\beta,\gamma,\delta,\epsilon}\mapsto z'^{\beta,\gamma,\delta,\epsilon}=z^{\beta,\gamma,\delta,\epsilon}+\tfrac{1}{2}z^\alpha+\{z^{\beta,\gamma,\delta,\epsilon},\Im_\ell H^\Lambda_\alpha\}_\Lambda,
\end{cases}
\end{align*}
and be decomposed as $W_\alpha^\Lambda=A_\alpha^\Lambda\circ B_\alpha^\Lambda$ with $A_\alpha^\Lambda,B_\alpha^\Lambda:R^E_\Lambda\to R^E_\Lambda$ given by
\begin{align*}
A_\alpha^\Lambda:
\begin{cases}
z^\alpha \mapsto 
z'{}^{\alpha}=-z^\alpha,
\cr
z^{\beta,\gamma,\delta,\epsilon}\mapsto
z'{}^{\beta,\gamma,\delta,\epsilon}=z^{\beta,\gamma,\delta,\epsilon}+\frac 1 2 z^\alpha,
\end{cases}
 B_\alpha^\Lambda:
\begin{cases}
z^\alpha \mapsto 
z'{}^{\alpha}=z^\alpha,
\cr
z^{\beta,\gamma,\delta,\epsilon}\mapsto
z'{}^{\beta,\gamma,\delta,\epsilon}=z^{\beta,\gamma,\delta,\epsilon}+\{z^{\beta,\gamma,\delta,\epsilon},\Im_\ell H_\alpha^\Lambda\}_{\Lambda}.
\end{cases}
\end{align*}
The transformation $B_\alpha^\Lambda$ is again a Hamiltonian flow, namely the one generated by $\Im_\ell H_\alpha$ via the gravitational Poisson structure.  The combinatorial transformation $A_\alpha^\Lambda$  sends the gravitational bivector $\pi_\Lambda$ for $\Gamma$ to the gravitational bivector $\pi_\Lambda'$ for  $\Gamma'=\Gamma_\alpha$, which means
$(A_\alpha^\Lambda)_*\pi_\Lambda=\pi_\Lambda'.$
When combined with Theorem \ref{thm:pentagon_shear_bend}, this generalises Theorem \ref{thm:pentagon_teich} from the context of Teichm\"uller space  to the moduli spaces of MGH Einstein spacetimes and proves that the mapping class group action on $\cG\cH_\Lambda(M)$ is symplectic.
\begin{theorem}
The Whitehead moves $W^\Lambda_\alpha:(R_\Lambda^E,\pi_\Lambda)\to (R_\Lambda^E,\pi_\Lambda')$ in \eqref{map_uni} are Poisson maps, and  the induced mapping class group action on $\cG\cH_\Lambda(M)$ is symplectic.
\end{theorem}
Another benefit of this decomposition of the Whitehead moves is that it has a geometrical interpretation.  For this, note that  the imaginary part of the complex dilogarithm is closely related to the Bloch-Wigner function $D(z)=\mathrm{Im}(\Li_2(z))+\log|z|\arg(1-z)$ which describes the volume of an ideal  hyperbolic tetrahedron in three-dimensional hyperbolic space. The volume of such  an ideal hyperbolic tetrahedron is given by the Bloch-Wigner function of the cross ratio of its vertices,  see for instance \cite{Zagier1990},  and there are similar results  for the spherical tetrahedra \cite{Murakami2005}. As a Whitehead move 
 corresponds to gluing an ideal tetrahedron on two adjacent ideal triangles in an ideal  triangulation, it is natural that the Hamiltonian generating this transformation is related to the volume of an ideal tetrahedron. In the context of three-dimensional Einstein manifolds, it seems plausible  that the Hamiltonian obtained for different values of $\Lambda$ could be related to the volumes of certain tetrahedra in three-dimensional Minkowski, anti de Sitter and de Sitter space.

\subsubsection*{Hamiltonians for the cotangent bundle}

We now show that a similar decomposition  of the Whitehead move is possible for the mapping class group action \eqref{gen_shear_pb} on the cotangent bundle $T^*\cT(S)$. For this, we construct  transformations $\tilde A^\Lambda_\alpha,\tilde B^\Lambda_\alpha:T^*\R^E\to T^*\R^E$, which are defined uniquely up to translations generated by the constraints  by the requirement  that the following diagrams commute
\begin{equation}
\xymatrix{
T^*\R^E \ar[dd]_{\widetilde A^\Lambda_\alpha}\ar[rr]^{\;\;\pi^\sharp_{WP}\quad} & &  R^E_\Lambda \ar[dd]^{A^\Lambda_\alpha}
\\ & & \\
T^*\R^E \ar[rr]_{\;\;\pi'{}^\sharp_{\!\!WP}\quad} & &  R^E_\Lambda
}
\qquad 
\xymatrix{T^*\R^E \ar[dd]_{\widetilde B^\Lambda_\alpha} \ar[rr]^{\;\;\pi^\sharp_{WP}\quad} &  &  R^E_\Lambda \ar[dd]^{B^\Lambda_\alpha} 
\\ & & \\
T^*\R^E 
\ar[rr]_{\;\;\pi^\sharp_{WP}\quad} & &  R^E_\Lambda.
}
\end{equation}
For this, we first pull-back $A^\Lambda_\alpha$ and $B^\Lambda_\alpha$ via the map $\pi_{WP}^\sharp:T^*\R^E\to R_\Lambda^E$ in \eqref{pipullb} 
and then compare the results with $\pi'{}^\sharp_{\!\!WP}$ and $\pi_{WP}^\sharp$, respectively. A similar computation to that of $A^\Lambda_\alpha$ and $B^\Lambda_\alpha$ together with the fact the map $\pi_{WP}^\sharp$ is Poisson shows that these transformations are given by
\begin{align*}
&\tilde A_\alpha^\Lambda:
\begin{cases}
x^\alpha \mapsto 
x'{}^\alpha=-x^\alpha,
\cr
x^{\beta,\gamma,\delta,\epsilon}\mapsto
x'{}^{\beta,\gamma,\delta,\epsilon}=x^{\beta,\gamma,\delta,\epsilon}+\frac{1}{2}x^\alpha,
\cr
p_\alpha \mapsto 
p_\alpha'=-p_\alpha+\frac{1}{2}(p_\beta+p_\gamma+p_\delta+p_\epsilon),
\cr
p_{\beta,\gamma,\delta,\epsilon}\mapsto p_{\beta,\gamma,\delta,\epsilon}'=p_{\beta,\gamma,\delta,\epsilon},
\end{cases}\\
\intertext{}
&\tilde B_\alpha^\Lambda:
\begin{cases}
x^\alpha \mapsto 
x'{}^\alpha=x^\alpha,
\cr
x^{\beta,\gamma,\delta,\epsilon}\mapsto
x'{}^{\beta,\gamma,\delta,\epsilon}=x^{\beta,\gamma,\delta,\epsilon}+\{x^{\beta,\gamma,\delta,\epsilon},\Im_\ell H_\alpha^\Lambda\circ \pi_{WP}^\sharp\}_{T^*},
\cr
p_\alpha \mapsto 
p_\alpha'=p_\alpha+\{p_\alpha,\Im_\ell H_\alpha^\Lambda\circ \pi_{WP}^\sharp\}_{T^*},
\cr
p_{\beta,\gamma,\delta,\epsilon}\mapsto p_{\beta,\gamma,\delta,\epsilon}'=p_{\beta,\gamma,\delta,\epsilon}.
\end{cases}
\end{align*}
This defines a decomposition of the Whitehead moves in \eqref{gen_shear_pb} as $\tilde W^\Lambda_\alpha=\tilde A^\Lambda_\alpha\circ \tilde B^\Lambda_\alpha$. The map  $\tilde B^\Lambda_\alpha$ is again a Hamiltonian flow, namely the one  generated by the Hamiltonian  $\Im_\ell H_\alpha^\Lambda\circ \pi_{WP}^\sharp:T^*\R^E \to\R$  via the cotangent bundle Poisson structure. The map
 $\tilde A^\Lambda_\alpha$ is again a linear combinatorial transformation, which does not depend on $\Lambda$ and maps  the cotangent bundle bivector $\pi_{T^*}$ for the graph $\Gamma$ to the cotangent bundle bivector $\pi'_{T^*}$ for the graph $\Gamma'=\Gamma_\alpha$, e.~g.~
$(\tilde A_\alpha^\Lambda)_*\pi_{T^*}=\pi'_{T^*}.$
This proves the following theorem which, combined with Theorem \ref{thm:mod_cot}, shows that the induced mapping class group actions on $T^*\cT(S)$  are symplectic for all values of $\Lambda$.
\begin{theorem}
The Whitehead moves $\tilde W^\Lambda_\alpha:(T^*\R^E,\pi_{T^*})\to (T^*\R^E,\pi'_{T^*})$ \eqref{gen_shear_pb} are Poisson maps, and  the induced mapping class group actions on $T^*\cT(S)$ are symplectic for all values of $\Lambda$.
\end{theorem}

\section*{Acknowledgements}
The authors thank Jean-Marc Schlenker for helpful discussions. This work was supported by the Emmy Noether research grant ME 3425/1-2 of the German Research Foundation (DFG).

\begin{appendix}

\section{Model spacetimes for 3d gravity}
\label{subsec:model}

In this appendix, we describe the model spacetimes of 3d gravity and their description in terms of the groups $\mathrm{PSL}(2,\R)$ and $\mathrm{PSL}(2,\C)$.
These model spacetimes are
three-dimensional  Minkowski space  $\mathrm{M}_3$ for $\Lambda=0$, anti-de Sitter space $\mathrm{AdS}_3$ for $\Lambda<0$, and de Sitter space $\mathrm{dS}_3$ for $\Lambda>0$. All of these model spacetimes are of constant curvature, which is given by the cosmological constant $\Lambda$,  and admit a simple description in terms of the Lie group $\mathrm{PSL}(2,\R)$ and its complexification $\mathrm{PSL}(2,\C)$.

\subsubsection*{The group $\mathrm{SL}(2,\R)$ and its Lie algebra}

In order to exhibit the similarities between the model spacetimes, it is helpful to first consider the Lie group $\mathrm{SL}(2,\R)$ and its Lie algebra.
For this, we introduce the following  basis of $\mathfrak{sl}(2,\R)$ 
\begin{align}\label{eq:jgen}
J_0=\tfrac 1 2 \left(\begin{array}{cc} 0 & -1\\ 1 & 0\end{array}\right)\qquad J_1=\tfrac 1 2 \left(\begin{array}{cc} 1 & 0\\ 0 & -1\end{array}\right)\qquad J_2=\tfrac 1 2 \left(\begin{array}{cc} 0 & 1\\ 1 & 0\end{array}\right),
\end{align}
which diagonalises the Killing form on $\mathfrak{sl}(2,\R)$ and relates it to the three-dimensional Minkowski metric $\eta=\text{diag}(-1,1,1)$
$$
\kappa(J_i,J_j)=\mathrm{Tr}(J_iJ_j)=\tfrac 1 2 \eta_{ij}.
$$

\subsubsection*{Minkowski geometry}
Three-dimensional  Minkowski space $\mathrm{M}_3$ is an affine space over the vector space $\R^3$ with the flat Lorentzian metric $\eta=\text{diag}(-1,1,1)$.
 Elements of $(\R^3,\eta)$ can be identified with the Lie algebra $\mathfrak{sl}(2,\R)$ via the map
$$\phi_0: \;(x^0,x^1,x^2)\mapsto X=2x^aJ_a=\left(\begin{matrix} x^1 & x^2-x^0 \cr x^2+x^0 & -x^1\end{matrix}\right),$$
such that the metric is given by minus the determinant
$$\eta_0(x,x)=-\det X=\tfrac 1 2 \tr X^2-\tfrac 1 2 (\tr X)^2
.$$
The  group of orientation and time orientation preserving isometries of $\mathrm{M}_3$ is the Poincar\'e group in three dimensions $G_0=\mathrm{ISO}(2,1)\cong\mathrm{PSL}(2,\R)\ltimes\mathfrak{sl}(2,\R)$, with the group multiplication
$$
(A,X)\cdot (B,Y)=(AB, X+AYA^\inv)\quad \text{where}\quad A,B\in \mathrm{PSL}(2,\R),\,X,Y\in\mathfrak{sl}(2,\R).
$$
With the identification $\mathfrak{sl}(2,\R)\cong \R^3$ from above,  its action on $\mathrm{M}_3$ is given by
$$(A,X)\cdot Y=AYA^\inv+X.$$
Two-dimensional hyperbolic space $\bbH^2$ embeds in Minkowski space as the space of future-oriented timelike unit vectors with the induced metric. In terms of the matrix realisation of $\mathrm{M}_3$ this can be described as the subspace of  matrices with determinant $1$. The subgroup of $G_0$ which preserves this embedding of $\bbH^2$ is  $\mathrm{PSL}(2,\R)\cong\{(A, 0): A\in\mathrm{PSL}(2,\R) \}$.

\subsubsection*{Anti-de Sitter geometry}

Three-dimensional anti-de Sitter space is defined as a quadric in $\R^4$
$$\mathrm{AdS}_3=\{(x^0,x^1,x^2,x^3)\in\R^4:\,-(x^0)^2+(x^1)^2+(x^2)^2-(x^3)^2=-1\}$$
with the Lorentzian metric induced by the flat pseudo-Riemannian metric $\text{diag}(-1,1,1,-1)$. 
Topologically, anti-de Sitter space is a product $\bbS^1\times\R^2$ presenting a closed timelike direction. This is somewhat irrelevant for our considerations since one can always unwrap such closed time direction by going to the universal covering space. On the other hand, it is sometimes also convenient to consider certain  quotient of anti-de Sitter space by the group $\Z_2$ and to work with its image in three-dimensional projective space
$$\mathrm{X}_{-1}=\{[x^0:x^1:x^2:x^3]\in\R\bbP^3:\,-(x^0)^2+(x^1)^2+(x^2)^2-(x^3)^2=-1\}.$$
As in the case of Minkowski space, three-dimensional anti-de Sitter space can be described in terms of the group $\mathrm{PSL}(2,\R)$.
The map
$$
\phi_{-1}:\,(x^0,x^1,x^2,x^3)\mapsto X=\left(\begin{array}{cc} x^3+x^1 & x^2-x^0\\
x^2+x^0 & x^3-x^1\end{array}\right)
$$
identifies $\mathrm{AdS}_3$ with the group  $\mathrm{SL}(2,\R)$ and $\mathrm{X}_{-1}$ with the group $\mathrm{PSL}(2,\R)$. In both cases the metric is the one induced by minus the determinant
$$\eta_{-1}(x,x)=-\det X=\tfrac{1}{2}\tr X^2-\tfrac 1 2 (\tr X)^2
.$$
The isometry group  of  $\mathrm{AdS}_3$ is the group  $\mathrm{SO}(2,2)\cong \mathrm{SL}(2,\R)\times \mathrm{SL}(2,\R)/\Z_2$, whose action on the  realisation above 
is given by the following action of $\mathrm{SL}(2,\R)\times\mathrm{SL}(2,\R)$
$$(A,B)\cdot X=A X B^\inv
.$$
As the kernel of this group action is $\{ (\mathds{1}, \mathds{1}), (-\mathds{1}, -\mathds{1})\}$ it induces an action of $\mathrm{SO}(2,2)$ on $\mathrm{AdS}_3$. Similarly the isometry group of $\mathrm{X}_{-1}$ is the group $G_{-1}=\mathrm{PSL}(2,\R)\times\mathrm{PSL}(2,\R)$.

Two-dimensional  hyperbolic space $\bbH^2$ embeds in $\mathrm{AdS}_3$ and in $\mathrm{X}_{-1}$ as a totally geodesic surface with the induced metric. In terms of the matrix realisation, $\bbH^2$ can be described as the  subspace characterised by the condition  $x^3=0$.  The subgroup of $G_{-1}$ which preserves $\bbH^2$ is given by the diagonal embedding $\mathrm{PSL}(2,\R)\cong\{(A,A):\, A\in\mathrm{PSL}(2,\R)\}$.

\subsubsection*{de Sitter geometry}
Three-dimensional de Sitter space  is also given as a quadric in $\R^4$
$$\mathrm{X}_{1}=\mathrm{dS}_3=\{(x^0,x^1,x^2,x^3)\in\R^4;\,-(x^0)^2+(x^1)^2+(x^2)^2+(x^3)^2=1\}$$
with the Lorentzian metric induced by the flat metric $\mathrm{diag}(-1,1,1,1)$. 
The map
$$\phi_{1}:\,(x^0,x^1,x^2,x^3)\mapsto {\mathrm i}X={\mathrm i}\left(\begin{matrix}{\mathrm i}x^3+x^1 & x^2-x^0\cr x^2+x^0 & {\mathrm i}x^3-x^1\end{matrix}\right)$$
identifies $\mathrm{dS}_3$ with a subset of $\mathrm{SL}(2,\C)$. The image of $\phi_{1}$ consists of those matrices of $\mathrm{SL}(2,\C)$ which are  invariant under the  involution
$$({\mathrm i}X)^\circ=J_0 ({\mathrm i}X)^\dag J_0^{-1}.$$
In this case, the induced metric is  given by the determinant 
$$\eta_1(x,x)=\det ({\mathrm i}X)=\tfrac{1}{2}\tr X^2-\tfrac 1 2 (\tr X)^2.
$$
The associated isometry group is the  Lorentz group in four dimensions $G_1=\mathrm{SO}(3,1)\cong \mathrm{PSL}(2,\C)$, which acts on the matrix realisation of $\mathrm{dS}_3$ via
$$
A\cdot ({\mathrm i}X)=A({\mathrm i}X)A^\circ.$$

Although the two-dimensional hyperbolic space $\bbH^2$ also embeds in de Sitter space as a totally geodesic surface, the relevant embedding of $\bbH^2$ is the one in the dual hyperbolic 3-space. The duality correspondence between $\mathrm{X}_1$ and $\bbH^3$ is given via the duality correspondence between 1-dimensional and 3-dimensional hyperplanes through the origin in $\R^{3,1}$. It maps points, geodesics and geodesic planes in $\mathrm{X}_1$, respectively, to geodesic planes, geodesics and points in $\bbH^3$. In terms of the matrix realisation above, this embedding of $\bbH^2$ can be described as the dual plane to the point $x^3=1$. The subgroup of $G_{1}$ which preserves $\bbH^2$ is then the subgroup $\mathrm{PSL}(2,\R)\subset \mathrm{PSL}(2,\C)$.

\section{The pentagon relation for the Whitehead move}
\label{app:flip}

\begin{theorem}
The Whitehead move \eqref{map_uni} satisfies the pentagon relation.
\end{theorem}
\begin{proof}
This follows by a direct computation from the expression \eqref{map_uni} for the Whitehead move in terms of shear-bending coordinates 
$z^\alpha:\cG\cH_\Lambda(M)\to R_\Lambda$.
Explicitly, the transformation of coordinates under the sequence of Whitehead moves in Figure \ref{fig:pentagon} is given by 
\begin{align*}
\left(\begin{array}{c} z^\alpha \\ z^\beta\\ z^\gamma \\ z^\delta \\z^\epsilon \\ z^\zeta\\ z^\eta\end{array}\right)
&\mapsto
\left(\begin{array}{c} z^\alpha +\log(1+e^{z^\zeta})\\ z^\beta-\log(1+e^{-z^\zeta})\\ z^\gamma\\z^\delta \\ z^\epsilon-\log(1+e^{-z^\zeta})\\ -z^\zeta\\ z^\eta+\log(1+e^{z^\zeta})
\end{array}\right)
\mapsto
\left(\begin{array}{c} z^\alpha+\log(1+e^{z^\zeta})\\ 
z^\beta-\log(1+e^{-z^\zeta})\\ 
z^\gamma+z^\eta-\log(1+e^{z^\eta}+e^{z^\zeta+z^\eta})+\log(1+e^{z^\zeta})\\ 
z^\delta+\log(1+e^{z^\eta}+e^{z^\zeta+z^\eta})\\
z^\epsilon+z^\zeta+z^\eta-\log(1+e^{z^\eta}+e^{z^\zeta+z^\eta})\\ 
-z^\zeta+\log(1+e^{z^\eta}+e^{z^\zeta+z^\eta})\\ 
-z^\eta-\log(1+e^{z^\zeta})
\end{array}\right)\\
&\mapsto
\left(\begin{array}{c} z^\alpha -\log(1+e^{z^\eta})+\log(1+e^{z^\eta}+e^{z^\eta+z^\zeta})\\ 
z^\beta+\log(1+e^{z^\eta})\\ 
z^\gamma-\log(1+e^{-z^\eta})\\ 
z^\delta+\log(1+e^{z^\eta}+e^{z^\zeta+z^\eta})\\
z^\epsilon+z^\zeta+z^\eta-\log(1+e^{z^\eta}+e^{z^\zeta+z^\eta})\\ 
z^\zeta-\log(1+e^{z^\eta}+e^{z^\zeta+z^\eta})\\ 
-z^\zeta+\log(1+e^{-z^\eta})
\end{array}\right)
\mapsto 
\left(\begin{array}{c} z^\alpha \\ 
z^\beta+\log(1+e^{z^\eta})\\ 
z^\gamma
-\log(1+e^{-z^\eta})\\ 
z^\delta+\log(1+e^{z^\eta})\\
z^\epsilon \\ 
-z^\eta\\ 
z^\zeta-\log(1+e^{-z^\eta})
\end{array}\right)
\mapsto
\left(\begin{array}{c} z^\alpha \\ z^\beta\\ z^\gamma\\ z^\delta \\z^\epsilon \\ z^\eta\\ z^\zeta
\end{array}\right).
\end{align*}
This proof is essentially  identical to the corresponding proof for shear coordinates on Teichm\"uller space and included  only   for the reader's convenience.
\end{proof}

\begin{figure}
\centering
\includegraphics[scale=0.4]{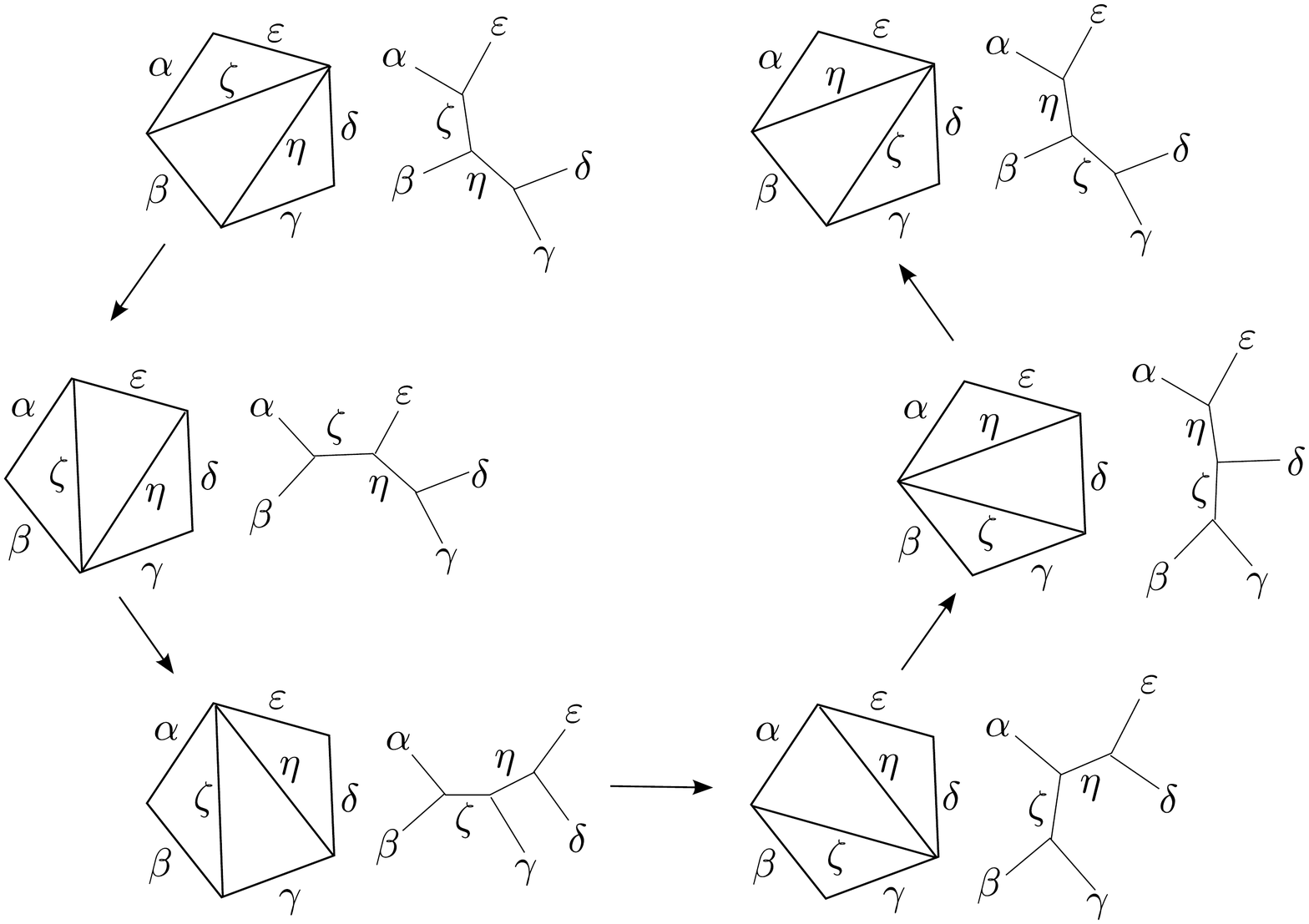}
\caption{The pentagon relation for the Whitehead move.}
\label{fig:pentagon}
\end{figure}

\end{appendix}


\end{document}